\documentclass[12pt]{amsart}

\usepackage{graphicx}
\usepackage{geometry}
\usepackage{amssymb}
\usepackage{setspace}
\usepackage[dvipsnames]{pstricks}
\usepackage{pst-node}
\usepackage{pst-slpe}
\usepackage{amsthm}
\usepackage{tipa}
\usepackage{protosem}
\usepackage{hieroglf}
\usepackage{subcaption}
\usepackage{graphicx}
\usepackage{amsmath}

\usepackage{array}
\usepackage{tikz}
\usepackage{pgfplots}
\pgfplotsset{compat=1.18}
\usetikzlibrary{arrows.meta,calc}
\usepackage{xcolor}
\usepackage{hyperref}
\definecolor{navy}{RGB}{31,78,121}

\pgfplotsset{
  expspace/.style={
    width=6.4cm, height=6.4cm,
    axis equal image,
    xmin=0, xmax=1, ymin=0, ymax=1,
    xtick={0,0.2,...,1}, ytick={0,0.2,...,1},
    minor tick num=1,
    xlabel={$\theta_{11}$},
    ylabel={$\theta_{22}$},
    xlabel style={font=\small},
    ylabel style={font=\small},
    tick label style={font=\footnotesize},
    grid=both,
    grid style={line width=0.2pt, draw=gray!30},
    major grid style={line width=0.3pt, draw=gray!40},
    axis line style={line width=0.6pt},
    clip=false,
  }
}

\DeclareMathOperator\conv{conv}

\usepackage{pst-tree}
\usepackage{amsmath}
\DeclareMathOperator*{\argmax}{argmax}

\usepackage{hyperref}
\hypersetup{
    colorlinks=true,
    linkcolor=blue,
    filecolor=magenta,      
    citecolor=blue,
}

\usepackage{cleveref}

\usepackage{amsfonts}
\usepackage{amsmath}
\usepackage{amssymb}
\usepackage{graphicx}
\usepackage{appendix}
\usepackage{multirow}
\usepackage{ulem}
\usepackage{setspace}
\usepackage[dvipsnames]{pstricks}
\usepackage{pst-node}
\usepackage{pst-slpe}
\usepackage{amsthm}
\usepackage{lmodern}
\usepackage[T1]{fontenc}
\usepackage{mathtools}
\usepackage{amsfonts}
\usepackage{multicol}
\usepackage{pdflscape}
\usepackage[latin9]{inputenc}
\usepackage{float}
\usepackage{amsmath}
\usepackage{amsthm}
\usepackage{xcolor}
\usepackage{amssymb}
\usepackage{geometry}
\usepackage{setspace}
\usepackage{graphicx}
\usepackage[authoryear]{natbib}
\usepackage{fancyhdr}
\usepackage{mathrsfs}

\usepackage{booktabs}
\usepackage{amssymb}
\usepackage{pifont}
\makeatletter
\newcommand{\norm}[1]{\left\lVert#1\right\rVert}

\usepackage{natbib}

  \theoremstyle{definition}
  \newtheorem{defn}{\protect\definitionname}

  \theoremstyle{plain}

  \theoremstyle{plain}
  
  \theoremstyle{plain}
  
  \theoremstyle{definition}
  
  \theoremstyle{plain}
  
  \theoremstyle{plain}
  \newtheorem{lemma}{Lemma}
  \theoremstyle{plain}
  
  \theoremstyle{plain}
  \newtheorem{proposition}{\protect\propositionname}
\theoremstyle{remark}

  \theoremstyle{definition}
  \newtheorem{example}{Example}

  \theoremstyle{definition}
  \newtheorem{assump}{Axiom}
  \newtheoremstyle{emptyplain}
    {}          
    {}          
    {\itshape}  
    {}          
    {\bfseries} 
    {.}         
    { }         
    {#3}        
\theoremstyle{emptyplain}

  \providecommand{\axiomname}{Axiom}
  \providecommand{\conjecturename}{Conjecture}
  \providecommand{\definitionname}{Definition}
  \providecommand{\corollaryname}{Corollary}
  \providecommand{\theoremname}{Theorem}
\providecommand{\propositionname}{Proposition}
\usepackage{xcolor,cancel}

\DeclareMathOperator\supp{supp}
\usepackage{empheq}
\usepackage[most]{tcolorbox}

\newtcbox{\mymath}[1][]{%
    nobeforeafter, math upper, tcbox raise base,
    enhanced, colframe=blue!30!black,
    colback=blue!30, boxrule=1pt,
    #1}

\begin{document}

\title[]{Agreement and Diversity in Interpretation}
\author[]{Francesco Bilotta$^\ast$}
\author[]{Luca Braghieri $^\S$}
\author[]{Collin Raymond$^\sharp$}
\author[]{Mark Whitmeyer$^\ddagger$}
\thanks{}
\thanks{$^\ast$ Bocconi University, E-mail: francesco.bilotta2@phd.unibocconi.it.}
\thanks{$^\S$ Bocconi University and NHH, CEPR, CESifo, and IGIER. E-mail: luca.braghieri@unibocconi.it}
\thanks{$^{\sharp}$ Cornell University, Email: collinbraymond@gmail.com.}
\thanks{$^{\ddagger}$ Arizona State University, Email: mark.whitmeyer@gmail.com.}
\date{July 2026.  We thank audiences at the University of Pennsylvania, RUD, BSGE Summer Forum, as well as Aislinn Bohren, Stephen Morris, Philipp Strack and Evan Piermont for helpful comments.}


\begin{abstract}
We study joint decision-making when agents agree on all primitives other than signal likelihoods. We propose a decision-theoretic measure of interpretive disagreement: a pair of subjective models is more agreeable than another if, uniformly across decision problems, it supports a larger set of signal-contingent plans that both agents weakly prefer ex-ante to the common reservation payoff. We show that this measure is prior independent and can be represented as an inclusion preorder over pairs of subjective models: each model in the more agreeable pair is a convex combination of the two models in the less agreeable pair. We then show that the measure's unique rotation-invariant scalar completion is cosine similarity. Applications show that greater agreement reduces speculative-trade wedges, expands a normalized version of the ex-ante Pareto frontier, and enlarges the set of single-model rationalizations. Our order is independent of Blackwell dominance and selects quadratic over KL-type Bregman divergences.
\end{abstract}

\maketitle
\section{Introduction}\label{sec:intro}

In many economic environments, people disagree about how to interpret a common piece of evidence. Doctors arrive at different conclusions from the same diagnostic test \citep{chan2022selection}, financial analysts who observe the same public information form different beliefs about a firm's future prospects \citep{kandel1995differential,bamber1999differential,bastianello2025mental}, and central bankers draw different inferences from the same macroeconomic indicators \citep{howes2026monetary,kaminski2026monetary}. The spread of AI systems that generate information through opaque, black-box processes is likely to make such disagreements even more prevalent, because people often interpret the same AI-generated output differently \citep{green2019principles,yu2024heterogeneity,eiermann2026child}.

Interpretive disagreements are consequential because they can impair joint decision-making even when agents' incentives are perfectly aligned. For instance, two doctors who both want their patient to recover may nevertheless disagree about treatment if they interpret the same test results differently. 

A measure of interpretive disagreement should satisfy two desiderata. First, it should be a property of the agents' subjective models, rather than of a particular decision problem, because the same evidence can inform many different decisions. The same diagnostic test, for example, can inform treatment, triage, and trial enrollment. Second, it should have economic content: when two interpretations are "closer", agents should find it easier to agree on a joint course of action.

In this paper, we propose a decision-theoretic measure of interpretive (dis)agreement that satisfies these two desiderata. To isolate the role of interpretation, we study agents who share the same prior, state space, action set, payoff function, reservation payoff, and Bayesian updating rule, but disagree about signal likelihoods.\footnote{We use the terms signal likelihood, signal structure, subjective model, and Blackwell experiment interchangeably.} We say that one pair of subjective models exhibits more agreement than another if it supports a larger set of signal-contingent plans that both agents are willing to accept ex-ante, relative to the problem's reservation payoff. The comparison yields a preorder over pairs of subjective models and, in the spirit of Blackwell's comparison of experiments \citep{blackwell1951comparison,blackwell1953equivalent}, is uniform across priors and payoff environments.

Our construction has a simple geometric representation. Each signal-contingent plan induces a surplus vector over state-signal pairs, measuring payoffs relative to the problem's fixed reservation payoff. Each agent evaluates this vector using her subjective joint distribution over states and signals, so each \textit{ex-ante} participation constraint is a half-space in the surplus space. The surplus vectors acceptable to both agents form the intersection of two half-spaces, a convex cone. We, therefore, say that one pair of models is more agreeable than another when its jointly acceptable cone contains the other's.

Our main result reveals that this cone-inclusion comparison has an especially simple representation. One pair of models supports a larger set of jointly acceptable signal-contingent plans if and only if each model in the more agreeable pair is a convex combination of the two models in the less agreeable pair. Thus, a comparison defined through \textit{ex-ante} participation constraints reduces to a prior-independent convexity criterion on signal structures themselves.

We then develop three applications that illustrate the economic content of the inclusion preorder. The first is speculative trade. A classic motivation for introducing belief disagreement to economic models is to explain how agents can rationally trade against one another despite no-trade results. We revisit this idea in a setting where agents observe the same evidence but interpret it differently. When agents disagree about signal likelihoods, they assign different values to the same state-signal-contingent contract. The gap between these valuations determines the interval of transfers under which the agents are willing to take opposite sides of the contract. We show that greater agreement shrinks this speculative-trade interval for every contract. In other words, the same order that expands the set of cooperative surplus plans reduces the scope for mutually acceptable speculative bets.

Our second application is to the \text{ex-ante} Pareto frontier of jointly-feasible payoffs. Greater agreement expands the set of mutually acceptable surplus vectors, but beliefs also determine how any given surplus plan is evaluated. We show that, after an appropriate normalization, the Pareto frontier generated by the more agreeable pair dominates the payoff set generated by the less agreeable pair. 

The third application is to rationalizing models: single Bayesian models under which an outside observer could interpret the agents' joint behavior as optimal. This captures the idea that joint decisions often need to be defensible as the product of a coherent common narrative, even when the agents privately disagree about how to interpret the evidence. We show that greater agreement enlarges the set of such narratives.

After highlighting the applications of the inclusion preorder, we turn one of its limitations, namely that it is incomplete. We show that the inclusion preorder admits a unique rotation-invariant scalar completion, given by the canonical statistical notion of cosine similarity. The rotation itself carries economic content, reflecting the fact that the Arrow-Debreu surplus space can be decomposed into directions along which the agents' evaluations are aligned, directions along which they disagree, and directions that are irrelevant for their participation constraints. Consequently, cosine similarity emerges as the canonical scalar measure of how \textit{close} two interpretations are to supporting the same jointly-acceptable plans.

We end our analysis with three comparisons that clarify the scope of our approach. First, we compare the inclusion preorder with Blackwell dominance and show that the two orders are not nested. Second, we compare our scalar completion with Bregman divergences, a family of statistical distances that includes the Kullback-Leibler divergence as a leading example \citep{kullback1951information,bregman1967relaxation}. Within the Bregman class, the completion of our inclusion preorder is equivalent to the quadratic divergence between normalized belief vectors, thereby distinguishing our measure from KL-type measures of belief disagreement. Finally, we compare our \text{ex-ante} approach with an interim version in which acceptability is evaluated after the signal is realized.

\subsection{Related Work.}
The dogmatic disagreement we consider in our paper relates to the literature on the common-prior assumption. \citet{morris1995common} provides an insightful survey of the role of common priors and emphasizes a familiar concern: once common priors are abandoned,  models of disagreement can become too flexible. We depart from the common-prior benchmark in a structured way. Agents share a common prior over payoff-relevant states, but disagree about signal likelihoods. This differentiates our work starkly from the papers establishing no-trade results, which require not only a common prior over payoff-relevant states, but that agents share a common model of how signals are generated, which is itself common knowledge.\footnote{See for instance the seminal no-trade results of \citet{milgrom1982information,sebenius1983don,morris1994trade}.}

A related recent literature studies learning and behavior under misspecified models, focusing on settings in which an agent's subjective model differs from the true model of the world.\footnote{Recent contributions study asymptotic behavior with endogenous actions and misspecification \citep{esponda_pouzo_2016_berknash,esponda_pouzo_2025_bnr,
esponda_pouzo_yamamoto_2021_abblm,fudenberg_lanzani_strack_2021_ecta}, incorrect beliefs about signal precision \citep{heidhues2021_convergence_te}, long-run beliefs under  misspecified signal-generating processes \citep{bohren_hauser_2021_ecta}, misspecified social learning \citep{frick2020_fragility_ecta,bohren2016_jet}, convergence under 
misspecified learning \citep{frick2023_belief_restud}, behavioral foundations of misspecification \citep{bohren2023behavioral}, and welfare comparisons for biased learning \citep{frick2024}. \citet{bohren_hauser_2025_are} survey the literature.} Our approach is different. We do not ask how behavior, beliefs, or welfare change when an agent's model differs from the truth. Instead, we take both agents' subjective models as primitives and ask how to measure the relative disagreement between them.

One way to interpret our participation constraints is that they ask when a committed signal-contingent decision rule has positive value relative to a fixed reservation payoff. This formulation is close to the decision-rule perspective in Wald's statistical decision theory \citep{wald1950statistical}, where rules are evaluated by their payoff or risk vectors and compared with other feasible rules. Beginning  with \citet{blackwell1951comparison,blackwell1953equivalent}, a large literature studies the value of information for a single decision maker.\footnote{Related work asks whether  information can be harmful outside expected utility \citep{wakker,hilton1990,safra1995,hill2020},and how information should be valued under non-Bayesian updating or ambiguity \citep{li2016,celen,bordoli2024,whitmeyer2023blackwell,von2024perils,deimen2025blackwell,
braghieri2026meaningful}. \citet{escude2025misperception} study when information reduces disagreement about priors, in a setting where agents agree on the experiment.} Our approach differs from these papers in that we do not ask whether information is valuable for a single decision maker. Instead, we ask which committed signal-contingent plans are deemed valuable by both agents under their respective subjective models, and use this joint acceptability criterion as a primitive for measuring disagreement.

Finally, our applications relate to trade and welfare under belief disagreement. The  speculative-trade application connects to work on heterogeneous beliefs in financial markets,  following \citet{harrison1978speculative}, and in bilateral trade, as in \citet{morris1994trade}.\footnote{Related work studies trade when belief disagreement is combined with ambiguity aversion, which can change agents' willingness to trade \citep{billot2000sharing,rigotti2005uncertainty,rigotti2008subjective}.} We isolate  disagreement about interpretation and show how its magnitude, as captured by our agreement order, governs the scope for speculative trade. A related welfare literature asks when, in environments with belief disagreement, one can make credible welfare comparisons \citep{brunnermeier2014welfare,gilboa2014no}. We sidestep these normative questions by taking each agent's subjective welfare assessment as given. Nevertheless, there is a useful connection: \citet{brunnermeier2014welfare} treat convex combinations of subjective beliefs as a set of ``reasonable'' beliefs, while convexification (endogenously) plays a central role in our analysis via our agreement order.

\vspace{0.5cm}

\noindent \textbf{Roadmap.} The remainder of the paper proceeds as follows. Section \ref{sec:model} presents the model and develops the geometric representation of beliefs as state-signal vectors. Section \ref{sec:inclusion} defines the cone-inclusion preorder and proves its convex-hull characterization. Section \ref{sec:app} develops our three applications: speculative trade, \text{ex-ante} Pareto frontiers, and rationalizing models. Section \ref{sec:cosine} establishes cosine similarity as the unique rotation-invariant completion of the inclusion order. Section \ref{sec:disc} connects the order to Blackwell's order and Bregman divergences and discusses interim participation. Section \ref{sec:con} concludes. Appendix A contains omitted proofs, and the Appendix B (online) contains expanded discussion.

\section{Model}\label{sec:model}
\subsection{Framework}\label{sec:framework}
We consider two agents $i\in I=\{1,2\}$ who face a common decision problem under uncertainty. The state of the world is $\omega\in\Omega$, where $\Omega$ is finite with $|\Omega|=n$. Neither agent observes $\omega$ directly, and the two agents share a common prior $p\in\Delta(\Omega)$.

A decision problem is a pair \((A,u)\), where \(A\) is a finite action set and \(u\colon A\times\Omega\to\mathbb R\) is a common utility function. To evaluate protocols, we also fix a scalar reservation payoff \(\bar u\), common to the agents. If no signal-contingent protocol is adopted, the agents' ex-ante benchmark payoff is \(\bar u\).\footnote{When useful, and in particular in Section~\ref{sec:rationalizing}, we assume that the reservation payoff can be represented by a feasible reservation action \(a^0\in A\) satisfying
\[
\sum_{\omega\in\Omega}p(\omega)u(a^0,\omega)=\bar u.
\]
The action \(a^0\) need not deliver \(\bar u\) in every state and need not be prior-optimal. If, instead, one requires \(\bar u\) to be generated by an action that is prior-optimal under \(p\), the full-space cone-inclusion criterion characterized in \Cref{prop:inclusion_char} remains sufficient for the same uniform behavioral implication, but is no longer necessary. Under that alternative, finite decision problems generate only a subcone \(K_p\subsetneq\mathbb R^{\Omega\times S}\), so the exact behavioral criterion would compare \(C(p_{\hat m},p_{\hat m'})\cap K_p\) and \(C(p_m,p_{m'})\cap K_p\). We develop this prior-optimal-action-benchmark variant in the online appendix (\Cref{app:payoff_vectors_decision_problems}).} To keep notation light, we suppress dependence on \(\bar u\) unless it is useful to make it explicit. Agents may randomize over actions, so for any belief \(q\in\Delta(\Omega)\) we define the set of optimal (possibly mixed) actions as
\[
A^*(q)\coloneqq\arg\max_{\alpha\in\Delta(A)}\mathbb E_{\omega\sim q}\left[U(\alpha,\omega)\right],
\]
where \(U(\alpha,\omega)\coloneqq\sum_{a\in A}\alpha(a)u(a,\omega)\).

Before choosing an action, the agents observe a common signal realization about the state of the world and update their beliefs according to Bayes' rule. The agents agree on the state space \(\Omega\), the prior \(p\), the finite space of possible signal realizations \(S\) (with \(\left|S\right|=k\)), and on which
signal \(s\in S\) is realized. The only potential disagreement is on the conditional distribution of signals given the state of the world.\footnote{This misspecification is dogmatic: the agents do not share a commonly understood model of uncertainty over signal structures that could reconcile their likelihoods and restore agreement.}

Formally, agent $i$ is characterized by a subjective joint distribution $p_i\in\Delta(\Omega\times S)$ whose $\Omega$-marginal coincides with the common prior: $p_i(\omega)=p(\omega)$ for all $\omega\in\Omega$. Given this joint distribution, each agent's subjective likelihood is defined implicitly by $p_i(\omega,s)=p(\omega) p_i(s\mid\omega)$. We further write $p_i(s)=\sum_{\omega\in\Omega}p_i(\omega,s)$ for the induced (subjective) distribution over signals, and $p_i(\omega\mid s)=p_i(\omega,s)/p_i(s)$ for agent $i$'s posterior whenever $p_i(s)>0$. Thus, the two agents may differ both in their induced unconditional distributions over signals, \(p_1(s)\) and \(p_2(s)\), and in their subjective posteriors after observing signal \(s\), \(p_1(\omega\mid s)\) and \(p_2(\omega\mid s)\). We denote the set of subjective signal structures \(m \colon \Omega \to \Delta(S)\) by $\mathcal M$.

Given a decision problem \((A,u)\) and a reservation payoff \(\bar u\), a \textit{signal-contingent protocol}, in short, a \textit{protocol}, is a mapping
\[
\sigma_{(A,u)}\colon S\to\Delta(A)
\]
that assigns a possibly randomized action to each signal. We evaluate protocols using each agent's \text{ex-ante} participation constraint: before observing the signal, an agent weakly prefers committing to the protocol to receiving the reservation payoff \(\bar u\).

\begin{defn}
Given a decision problem \((A,u)\) and a reservation payoff \(\bar u\), a protocol \(\sigma_{(A,u)}\) satisfies the ex-ante participation constraint for agent \(i\in I\) if \[\sum_{\omega\in\Omega}\sum_{s\in S} p_i(\omega,s) \sum_{a\in A}\sigma_{(A,u)}(a\mid s) u(a,\omega) - \bar u \geq 0.\]
\end{defn}

\subsection{A Geometric Interpretation}\label{sec:geometry}

The \text{ex-ante} joint decision problem admits a convenient geometric representation in the spirit of Arrow-Debreu. We represent (i) each agent's subjective beliefs as a vector of state-signal probabilities and (ii) each protocol as a vector of state-signal contingent surplus relative to the reservation payoff.\footnote{Fix \(\bar u\). Add one action \(a^s\) for each signal \(s\in S\), with \(u(a^s,\omega)=\bar u+x(\omega,s)\), and let the protocol choose \(a^s\) after signal \(s\). If a feasible reservation representation is desired, add a reservation action \(a^0\) satisfying \(\sum_{\omega\in\Omega}p(\omega)u(a^0,\omega)=\bar u\); for example, one may take \(u(a^0,\omega)=\bar u\) for every \(\omega\). The induced surplus vector is \(x\).} With these representations, \text{ex-ante} participation constraints become linear inequalities in a Euclidean space. Their intersection is the
geometric object behind our measures of disagreement.

\medskip

\noindent\textbf{Beliefs as vectors.}
To write beliefs as vectors, we identify $\mathbb{R}^{nk}$ with the space of real-valued arrays indexed by $\Omega\times S$.\footnote{Any fixed ordering of $\Omega\times S$ induces such an identification.} For each agent $i\in I$, define the belief vector
\[
p_i
\equiv
\left(p_i(\omega_1,s_1),\ldots,p_i(\omega_n,s_1),\,p_i(\omega_1,s_2),\ldots,p_i(\omega_n,s_k)\right)
\in \mathbb{R}^{nk}.
\]
Thus, $p_i$ stacks agent $i$'s subjective joint probabilities over state-signal pairs.

\medskip

\noindent\textbf{Protocols as surplus vectors.}
Fix a decision problem $(A,u)$ and a protocol $\sigma\colon S\to\Delta(A)$. For each signal $s$ and state $\omega$, the protocol induces the expected payoff
\[
U(\sigma(\cdot\mid s),\omega)
\equiv
\sum_{a\in A}\sigma(a\mid s) u(a,\omega).
\]
We measure payoffs relative to the reservation payoff and define the associated surplus array
\[
x(\omega,s)\equiv U(\sigma(\cdot\mid s),\omega)-\bar u.
\]
Stacking $x(\omega,s)$ across $(\omega,s)\in\Omega\times S$ in the same order as $p_i$ yields a surplus vector
\[
x
\equiv
\left(x(\omega_1,s_1),\ldots,x(\omega_n,s_1),\,x(\omega_1,s_2),\ldots,x(\omega_n,s_k)\right)
\in \mathbb{R}^{nk}.
\]

Although payoffs depend only on the state \(\omega\), protocols are signal-contingent, and the \text{ex-ante} participation constraint averages the induced payoff differences over state-signal realizations. Consequently, the relevant belief object is the joint distribution on \(\Omega\times S\). When agents disagree about likelihoods, they disagree precisely about these state-signal probabilities. Representing surplus vectors in \(\mathbb R^{\Omega\times S}\) keeps track of that disagreement while making participation constraints linear. As we will see, it is also convenient to view \(\mathbb R^{\Omega\times S}\) as an Arrow-Debreu contract space indexed by state-signal pairs, where a primitive security \(\mathbf{1}(\omega,s)\) pays one unit of surplus
if and only if state \(\omega\) obtains and signal \(s\) is realized and a general surplus vector \(x\in\mathbb{R}^{\Omega\times S}\) is a portfolio of such securities, evaluated by agent \(i\) as \(p_i\cdot x\).

\medskip

\noindent\textbf{Participation constraints as half-spaces.}
Agent \(i\)'s \text{ex-ante} participation constraint for protocol \(\sigma\) is equivalent to
\[
p_i\cdot x
=
\sum_{\omega\in\Omega}\sum_{s\in S} p_i(\omega,s)\,x(\omega,s)
=
\sum_{\omega\in\Omega}\sum_{s\in S} p_i(\omega,s)U(\sigma(\cdot\mid s),\omega)-\bar u
\ge 0,
\]
where the dot product is taken in $\mathbb{R}^{nk}$. 

Viewing \(x\) as an arbitrary element of \(\mathbb R^{nk}\), the condition \(p_i\cdot x\ge 0\) describes a closed half-space with normal \(p_i\). We refer to a surplus vector \(x\) that satisfies the \text{ex-ante} participation constraint for agent
\(i\) as \textit{acceptable} for agent \(i\). Accordingly, the set of surplus vectors that are acceptable to both agents is the intersection
\[
C(p_1,p_2)
\equiv
\left\{x\in\mathbb{R}^{nk}\colon \ p_1\cdot x\ge 0, p_2\cdot x\ge 0\right\},
\]
which is a convex cone containing the origin.\footnote{The origin is the reservation-normalized zero vector: a payoff equal to \(\bar u\) at every state-signal pair.} Since \(p_i(\omega,s)\ge 0\) for all \((\omega,s)\), both normals lie
in the positive orthant. We refer to a surplus vector \(x\) that lies in this cone, and, hence, satisfies both agents' \text{ex-ante} participation constraints, as \textit{jointly acceptable}.

The relative orientation of $p_1$ and $p_2$ determines how ``large'' the cone of jointly acceptable surplus vectors is. If $p_1=p_2$ (agents agree on the signal structure), then $C(p_1,p_2)$ is just a single half-space. As the two belief vectors become less aligned, the intersection shrinks. The following lemma formalizes the intuition:

\begin{lemma}\label{lem:extreme}
    If agents fully agree on the signal structure, the cone of acceptable expected utility vectors is maximized (becomes a half-space). If their supports are disjoint, the cone is minimized.
\end{lemma}
If \(p_1=p_2\), then \(C(p_1,p_2) = \left\{x \colon \ p_1\cdot x\ge 0\right\}\), so the jointly acceptable cone is a single half-space. If \(p_1\cdot p_2=0\), equivalently, if the supports of \(p_1\) and \(p_2\) are disjoint, then the two normals are orthogonal. Since both normals lie in the positive orthant, this is the largest possible angle between them. In this case, joint acceptability separates across supports: each agent evaluates surplus only on events to which the other assigns zero probability. Thus, gains on one agent's support cannot compensate losses on the other's support and the cone is minimized.

\medskip
\noindent\textbf{From surplus vectors to an order over beliefs.}
A key advantage of this representation is that it cleanly separates beliefs from the decision problem. Every decision problem \((A,u)\) protocol \(\sigma\)  induces a surplus vector \(x\in\mathbb{R}^{nk}\) (recall that a decision problem includes the reservation utility). For any fixed decision problem, the set of attainable surplus vectors is restricted by the available actions and payoffs. But the belief-dependent part of joint acceptability is always captured by the same cone \(C(p_1,p_2)\). With a reservation-payoff benchmark, the ambient surplus space is the behaviorally relevant domain. That is, for every \(x\in\mathbb R^{\Omega\times S}\), there is a finite decision problem and protocol that generate \(x\) as surplus relative to the reservation payoff.\footnote{Under the prior-optimal-action benchmark variant in Appendix \ref{app:payoff_vectors_decision_problems} this is no longer true, which is why the result is weaker there; specifically not every surplus vector maps to a decision problem (where the reservation utility is prior optimal).}  

This motivates our approach. Rather than measuring disagreement within a particular decision problem, we compare belief pairs by comparing their cones of jointly acceptable surplus vectors. A pair exhibits more agreement when it renders a weakly larger set of
surplus vectors jointly acceptable.

Finally, each belief vector is generated by the common prior \(p\) together with a subjective signal structure. For \(m\in\mathcal M\), write \(p_m(\omega,s)\equiv p(\omega)m(s\mid\omega)\), so that a pair of subjective signal structures \((m,m')\) induces the jointly acceptable cone \(C(p_m,p_{m'})\). Our next section defines the inclusion preorder by comparing these cones.

\section{Measuring Disagreement}\label{sec:inclusion}

Our notion of agreement is based directly on the cones of \text{ex-ante} joint acceptability introduced in Section~\ref{sec:geometry}. The idea is simple: one pair of experiments exhibits \emph{weakly
more agreement} than another if every surplus vector that is jointly acceptable under the latter is also jointly acceptable under the former.

Formally, fix a common prior $p\in\Delta(\Omega)$ and two experiments $m,m'\in\mathcal M$. Recall that each experiment $m$ induces a joint distribution $p_m\in\Delta(\Omega\times S)$ via $p_m(\omega,s)=p(\omega)m(s\mid\omega)$, and that the corresponding jointly acceptable cone is
\[
C(p_m,p_{m'})
\equiv
\{x\in\mathbb{R}^{\Omega\times S}\colon\ p_m\cdot x\ge 0, p_{m'}\cdot x\ge 0\}.
\]

\begin{defn}[Inclusion preorder]\label{def:inclusion}
We say that $(m,m')$ \textit{exhibits weakly more agreement in the inclusion order than $(\hat m,\hat m')$ at prior $p$},
\((m,m')\succeq^{I}_p(\hat m,\hat m')\), if \(C(p_{\hat m},p_{\hat m'}) \subseteq C(p_m,p_{m'})\).
\end{defn}

By construction, if \((m,m')\succeq^I_p(\hat m,\hat m')\), then for \textit{every} decision problem \((A,u)\), every reservation payoff \(\bar u\), and every protocol \(\sigma\), whenever \(\sigma\) is \textit{ex-ante} jointly acceptable for \((\hat m,\hat m')\) at prior \(p\), it is also \textit{ex-ante} jointly acceptable for \((m,m')\) at the same prior. Conversely, because finite decision problems span all of \(\mathbb R^{\Omega\times S}\) once payoffs are measured relative to a reservation payoff, this uniform behavioral implication implies cone inclusion. Thus, higher agreement means weakly fewer \textit{ex-ante} contracting frictions uniformly across decision problems.

The relation $\succeq^{I}_p$ is generally incomplete. It is a preorder on pairs of experiments (distinct pairs can induce the same cone), and it induces a partial order on the induced cones of jointly-acceptable surplus vectors.\footnote{The reason why $\succeq^I_p$ is not a partial order on model pairs is that, given $m\neq m'$, we have $(m,m')\sim^I(m',m)$. Hence, antisymmetry fails. Nonetheless, a partial order on models remains indeed defined, if we identify pairs of models up to permutation.}

For prior-independent comparisons, we use the following definition.

\begin{defn}[Uniform inclusion preorder]\label{def:uniform_inclusion}
We say that \textit{\((m,m')\) exhibits uniformly weakly more agreement in the inclusion order than
\((\hat m,\hat m')\)}, \((m,m')\succeq^{I}(\hat m,\hat m')\), if and only if for all full-support \(p\), \((m,m')\succeq^{I}_p(\hat m,\hat m')\).
\end{defn}

The next result turns this geometric definition into a simple convexity test. The acceptable-surplus cone is defined in the payoff space \(\mathbb R^{\Omega\times S}\), but its dual is the cone spanned by the two belief vectors that define the agents' participation constraints. Hence making the jointly acceptable cone larger is equivalent to making the cone generated by the two normals smaller. Because the normals are probability vectors, this dual cone comparison reduces to a line-segment comparison: the more agreeable pair must lie inside the convex hull of the less agreeable pair.

\begin{proposition}\label{prop:inclusion_char}\label{prop:incl_prior_indep} Fix models \(m,m',\hat m,\hat m'\in\mathcal M\). For any full-support prior \(p\), \((m,m')\succeq^{I}_p(\hat m,\hat m')\)
if and only if there exist \(\alpha,\alpha'\in[0,1]\) such that
\[
p_m=\alpha p_{\hat m}+(1-\alpha)p_{\hat m'},
\qquad \text{and} \qquad
p_{m'}=\alpha' p_{\hat m}+(1-\alpha')p_{\hat m'}.
\]
Moreover, the following are equivalent:
\begin{enumerate}
    \item\label{prop1item1} There exists a full-support prior \(p\in\Delta(\Omega)\) such that \((m,m')\succeq^{I}_p(\hat m,\hat m')\).
    \item\label{prop1item2} \((m,m')\succeq^{I}(\hat m,\hat m')\).
    \item\label{prop1item3} There exist \(\alpha,\alpha'\in[0,1]\) such that
    \(m=\alpha\hat m+(1-\alpha)\hat m'\) and \(m'=\alpha'\hat m+(1-\alpha')\hat m'\).
\end{enumerate}
\end{proposition}

\Cref{prop:inclusion_char} gives the inclusion order a clear interpretation. A pair \((m,m')\) is more agreeable than \((\hat m,\hat m')\) if and only if both \(m\) and \(m'\) lie on the line segment between \(\hat m\) and \(\hat m'\). Thus the less agreeable pair \((\hat m,\hat m')\) brackets the more agreeable pair. Moving up the agreement order is therefore literally an averaging operation. Each agent's model in the more agreeable pair is a convex combination of the two models in the less agreeable pair, so neither agent holds an interpretation that is more extreme than the disagreement already present in the benchmark pair.

The second part of the proposition shows that this comparison is intrinsic to the experiments themselves. Although the definition of \(\succeq^I_p\) is written in terms of prior-weighted joint distributions \(p_m(\omega,s)=p(\omega)m(s\mid \omega)\), the full-support prior only rescales rows of the experiment by positive factors. These factors cancel from the 
convex-hull condition. Hence, if the comparison holds at one full-support prior, it holds at every full-support prior. The inclusion order is therefore not a property of a particular prior or decision problem; it is a prior-free partial order over pairs of subjective signal structures.

This characterization also makes clear why the order is incomplete. Two pairs are comparable only when the two experiments in one pair both lie on the line segment generated by the other pair. When the corresponding segments in experiment space point in different directions, neither pair brackets the other, and the inclusion order leaves them unranked.

We now illustrate our results in a simple environment.  

\begin{example}\label{ex:basic}
We fix \(|\Omega|=|S|=2\), write \(\Omega=\{\omega_1,\omega_2\}\),
\(S=\{s_1,s_2\}\), and impose a uniform prior
\(p(\omega_1)=p(\omega_2)=\tfrac{1}{2}\). In this case, an experiment is a
\(2\times 2\) row-stochastic matrix
\[
  m = \begin{pmatrix} \theta_{11} & \theta_{12} \\ \theta_{21} & \theta_{22}
      \end{pmatrix},
  \qquad \text{with} \qquad
  \theta_{i1}+\theta_{i2}=1,
\]
and where \(\theta_{ij}=m(s_j\mid\omega_i)\). The induced joint distribution is
\(p_m(\omega_i,s_j)=\tfrac{1}{2}\theta_{ij}\).

Consider first the pair \((A,A')\), where
\[
  A = \begin{pmatrix} 0.7 & 0.3 \\ 0.4 & 0.6 \end{pmatrix},
  \qquad \text{and} \qquad
  A' = \begin{pmatrix} 0.5 & 0.5 \\ 0.5 & 0.5 \end{pmatrix}.
\]
Experiment \(A\) is (partially) discriminating, while \(A'\) is uninformative. Now define the identical
pair
\[
  C = C' = \frac{1}{2}A + \frac{1}{2}A'
  = \begin{pmatrix} 0.60 & 0.40 \\ 0.45 & 0.55 \end{pmatrix},
\]
and the intermediate pair
\[
  B = \frac{1}{2}A+\frac{1}{2}C
    = \begin{pmatrix} 0.65 & 0.35 \\ 0.425 & 0.575 \end{pmatrix},
  \qquad \text{and} \qquad
  B' = \frac{1}{2}A'+\frac{1}{2}C'
     = \begin{pmatrix} 0.55 & 0.45 \\ 0.475 & 0.525 \end{pmatrix}.
\]

By \Cref{prop:inclusion_char},
\[
(C,C') \succeq^I (B,B') \succeq^I (A,A').
\]
Indeed,
\[
C=C'=\tfrac12 B+\tfrac12 B',
\qquad
B=\tfrac34 A+\tfrac14 A',
\qquad \text{and} \qquad
B'=\tfrac14 A+\tfrac34 A'.
\]

To illustrate the incompleteness of the inclusion preorder, consider
\[
  D = \begin{pmatrix} 0.8 & 0.2 \\ 0.2 & 0.8 \end{pmatrix},
  \qquad \text{and} \qquad
  D' = \begin{pmatrix} 0.7 & 0.3 \\ 0.3 & 0.7 \end{pmatrix}.
\]
Both \(D\) and \(D'\) are symmetric experiments, lying on the main diagonal of the
\((\theta_{11},\theta_{22})\) plane. The pair \((D,D')\) is incomparable with
\((A,A')\), \((B,B')\), and \((C,C')\). For example, \(D\in\conv\{A,A'\}\) would require
\(0.8=\alpha 0.7+(1-\alpha)0.5\), so \(\alpha=3/2\notin[0,1]\). Conversely, \(A\in\conv\{D,D'\}\) would require
\(\alpha=0\) from the \((\omega_1,s_1)\) entry, but then the \((\omega_2,s_1)\) entry would
require \(0.4=0.3\), a contradiction. The remaining checks are analogous.

Figure~\ref{fig:fig1} provides a two-dimensional representation of the example. We represent each experiment with \((\theta_{11},\theta_{22})\): the probability of \(s_1\) in state \(\omega_1\) and the probability of \(s_2\) in state \(\omega_2\). Dashed segments connect the two members of
\((A,A')\) and \((B,B')\), and the pair \((C,C')\) is a single point (since \(C=C'\)). 
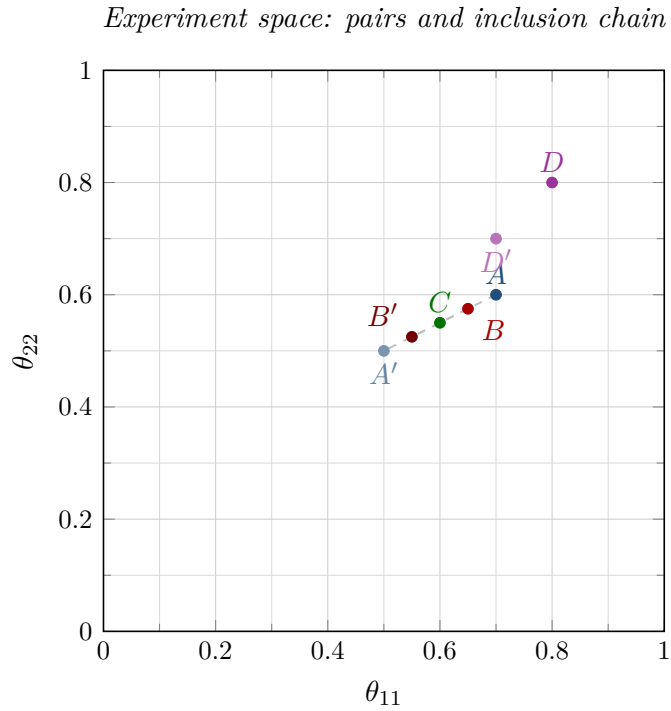
\begin{figure}[h]
  \centering
  \begin{tikzpicture}
\begin{axis}[
  width=9cm, height=9cm,
  axis equal image,
  xmin=0, xmax=1, ymin=0, ymax=1,
  xtick={0,0.2,...,1}, ytick={0,0.2,...,1},
  minor tick num=1,
  xlabel={$\theta_{11}$},
  ylabel={$\theta_{22}$},
  xlabel style={font=\small},
  ylabel style={font=\small},
  tick label style={font=\footnotesize},
  grid=both,
  grid style={line width=0.2pt, draw=gray!30},
  major grid style={line width=0.3pt, draw=gray!40},
  axis line style={line width=0.6pt},
  clip=false,
  title={\small Experiment space: pairs and inclusion chain},
  title style={font=\small\itshape, at={(0.5,1.02)}},
]

\addplot[dashed, gray!50, line width=0.7pt] coordinates {(0.70,0.60)(0.50,0.50)};
\addplot[dashed, gray!50, line width=0.7pt] coordinates {(0.65,0.575)(0.55,0.525)};

\addplot[navy,        mark=*, mark size=2pt, only marks] coordinates {(0.70,0.60)};
\addplot[navy!60,     mark=*, mark size=2pt, only marks] coordinates {(0.50,0.50)};
\addplot[red!65!black,mark=*, mark size=2pt, only marks] coordinates {(0.65,0.575)};
\addplot[red!45!black,mark=*, mark size=2pt, only marks] coordinates {(0.55,0.525)};
\addplot[green!45!black, mark=*, mark size=2pt, only marks] coordinates {(0.60,0.55)};
\addplot[violet!80,   mark=*, mark size=2pt, only marks] coordinates {(0.80,0.80)};
\addplot[violet!55,   mark=*, mark size=2pt, only marks] coordinates {(0.70,0.70)};

\node[font=\small\bfseries, navy, anchor=south]
    at (axis cs:0.700,0.600) {$A$};
\node[font=\small\bfseries, navy!70, anchor=north]
    at (axis cs:0.500,0.500) {$A'$};
\node[font=\small\bfseries, red!65!black, anchor=north west]
    at (axis cs:0.655,0.575) {$B$};
\node[font=\small\bfseries, red!45!black, anchor=south east]
    at (axis cs:0.545,0.525) {$B'$};
\node[font=\small\bfseries, green!45!black, anchor=south]
    at (axis cs:0.600,0.550) {$C$};
\node[font=\small\bfseries, violet!80, anchor=south]
    at (axis cs:0.800,0.800) {$D$};
\node[font=\small\bfseries, violet!55, anchor=north]
    at (axis cs:0.700,0.700) {$D'$};

\end{axis}
\end{tikzpicture}
  \caption{Experiment space $[0,1]^2$ with axes $\theta_{11}$ (horizontal) and
    $\theta_{22}$ (vertical).  The inclusion chain is
    $(A,A')\prec_I(B,B')\prec_I(C,C')$.}
  \label{fig:fig1}
\end{figure}
\end{example}

\section{Applications}\label{sec:app}

This section presents three applications of the geometry developed above. First, we show that greater disagreement expands the scope for mutually acceptable speculative trade: agents who interpret the same signal differently may assign different values to the same state-signal-contingent side bet. Second, we study how changes in disagreement affect the \text{ex-ante} Pareto frontier and show that, after a suitable normalization, greater agreement expands the frontier. Third, we relate disagreement to rationalizing models, namely single Bayesian models under which an outside observer could rationalize the agents' jointly acceptable behavior as optimal. These applications show that the inclusion order has economic content beyond the geometry used to define it.

\subsection{Speculative Trade}\label{sec:spec} Recall that a surplus vector \(x\in\mathbb{R}^{\Omega\times S}\) is an Arrow-Debreu portfolio over state-signal pairs, with an \textit{ex-ante} value of \(p_m\cdot x\) under model \(m\). The cone \(C(p_m,p_{m'})\) describes when the same surplus vector is acceptable to both agents. This subsection asks a different, but closely related, question: if \(x\) is instead understood as a zero-net-supply side bet, with one agent receiving \(x\) and the other receiving \(-x\), which transfers make both agents willing to trade?\footnote{The trade is speculative in the sense that the payoff environment, the prior, and the surplus vector \(x\) are common; only the \textit{ex-ante} evaluations \(p_m\cdot x\) and \(p_{m'}\cdot x\) may differ.}

To wit, suppose agents can make transfers in the same utility units as the surplus vector. If the agent with model \(m\) receives \(x\) and pays transfer \(q\), the two agents' \textit{ex-ante}
gains are
\[
p_m\cdot x-q
\qquad\text{and}\qquad
q-p_{m'}\cdot x.
\]
Accordingly, this trade is acceptable if and only if
\(p_{m'}\cdot x\leq q\leq p_m\cdot x\). Naturally, if the inequality between valuations is reversed, the agents reverse sides of the bet.
In sum, regardless of which agent values \(x\) more, the set of mutually acceptable transfers is the interval between the two \textit{ex-ante} valuations:
\[
\mathcal{T}_p(x;m,m')
\coloneqq
\left[
\min\{p_m\cdot x,p_{m'}\cdot x\},
\max\{p_m\cdot x,p_{m'}\cdot x\}
\right].
\]

The interval \(\mathcal{T}_p(x;m,m')\) is precisely a disagreement interval. If \(p_m=p_{m'}\), the two agents assign the same value to every side bet \(x\), so the interval collapses to a point.
More generally,
\[
\mathcal{T}_p(x;m,m')
=
\left[
\frac{p_m+p_{m'}}{2}\cdot x-\frac{1}{2}|p_m\cdot x-p_{m'}\cdot x|,
\frac{p_m+p_{m'}}{2}\cdot x+\frac{1}{2}|p_m\cdot x-p_{m'}\cdot x|
\right].
\]
The midpoint is the average \textit{ex-ante} valuation of \(x\), and the radius is one half of the absolute valuation gap. The interval is wide precisely when the two agents disagree sharply about the value of \(x\).

The valuation gap also decomposes into two components, mirroring the prospective/retrospective distinction in \citet{bohren2023behavioral}. When the relevant signal marginals are positive under both models, writing \(x_s(\omega)\coloneqq x(\omega,s)\), we have
\[
\begin{aligned}
p_m\cdot x-p_{m'}\cdot x
&=\sum_{s\in S}\left(p_m(s)-p_{m'}(s)\right)
\frac{p_m(\cdot\mid s)+p_{m'}(\cdot\mid s)}{2}\cdot x_s\\
&\quad+\sum_{s\in S}\frac{p_m(s)+p_{m'}(s)}{2}
\left(p_m(\cdot\mid s)-p_{m'}(\cdot\mid s)\right)\cdot x_s.
\end{aligned}
\]
The first term is prospective speculative disagreement: agents disagree about how likely a signal is. The second term is retrospective speculative disagreement: conditional on the signal, agents disagree about what it means. 

The next proposition provides the speculative-trade meaning of the inclusion order.

\begin{proposition}\label{prop:spec} Take any full-support prior \(p\). Then,
\[(m,m') \succeq^I_p(\hat m,\hat m') \qquad
\Longleftrightarrow \qquad
\mathcal{T}_p\left(x;m,m'\right)\subseteq\mathcal{T}_p\left(x;\hat m,\hat m'\right) \quad \forall \ x \in \mathbb{R}^{\Omega\times S}.\]
\end{proposition}

\Cref{prop:spec} shows that more agreement enlarges the cone of jointly acceptable surplus
vectors, but shrinks the speculative-transfer interval for every side bet. The two comparisons
ask opposite questions. The cone asks which surplus vectors both agents value weakly
positively; the transfer interval asks how far apart their valuations of a fixed zero-net-supply
side bet are.

By combining \Cref{prop:spec} with \Cref{prop:inclusion_char}, we can also interpret the uniform inclusion preorder as a prior-independent nesting order on \textit{ex-ante} mutually agreeable transfer intervals. If
\((m,m')\succeq^I(\hat m,\hat m')\), then for every full-support prior \(p\) and every
\(x\in\mathbb{R}^{\Omega\times S}\),
\[
\mathcal{T}_p(x;m,m')
\subseteq
\mathcal{T}_p(x;\hat m,\hat m').
\]
Conversely, if this nesting holds for some full-support prior and every \(x\), then
\((m,m')\succeq^I(\hat m,\hat m')\).

In order to illustrate this result, we can use the experiments from Example \ref{ex:basic}.  The width of the jointly acceptable transfer interval for surplus vector $x$ is
$|(p_m-p_{m'})\cdot x|$.  Since $p_B-p_{B'}=\tfrac12(p_A-p_{A'})$,
\[
  |(p_A-p_{A'})\cdot x| \geq |(p_B-p_{B'})\cdot x| \geq 0 =
  |(p_C-p_{C'})\cdot x| \quad\text{for every }x.
\]

For the surplus vector \(x=(1,0,-1,0)^\top\), which pays \(+1\) at \((\omega_1,s_1)\), \(-1\) at \((\omega_1,s_2)\),
and zero otherwise, we have
\[
p_A\cdot x=0.20,\qquad \text{and} \qquad p_{A'}\cdot x=0,
\]
so \(\mathcal{T}_p(x;A,A') = \left[0,0.20\right]\), with width \(0.20\). For \((B,B')\),
\(p_B\cdot x=0.15\) and \(p_{B'}\cdot x=0.05\), with width \(0.10\). For \((C,C')\),
\(p_C\cdot x = 0.10 = p_{C'}\cdot x\), which means there is no speculative trade at all.

In short, agreement has opposite implications for cooperation and speculation. On the one hand, greater agreement makes it easier to find surplus vectors that both agents value positively, but harder to find zero-net-supply bets that both agents are willing to trade. On the other hand, greater disagreement expands the scope for speculative trade because it increases the set of contracts over which agents assign different \textit{ex-ante} values.

\subsection{Agreement and \text{ex-ante} Pareto Frontiers}\label{sec:pareto}

We now use the cone geometry to study how likelihood disagreement shapes \text{ex-ante} Pareto frontiers within a fixed decision problem. Fix a decision problem \((A,u)\) and a reservation payoff \(\bar u\). Suppressing dependence on \(\bar u\), let \(X(A,u)\subset \mathbb R^{nk}\) denote the set of surplus vectors induced by signal-contingent protocols for \((A,u)\), measured relative to \(\bar u\).

For a pair \((m,m')\), define the set of jointly acceptable surplus vectors in decision problem
\((A,u)\) as
\[
K(m,m';A,u)
\equiv
X(A,u)\cap C(p_m,p_{m'})
\subseteq \mathbb R^{nk},
\]
and define the induced payoff set in \(\mathbb R^2\) as
\[
J(m,m';A,u)
\equiv
\left\{
(p_m\cdot x, p_{m'}\cdot x)\in\mathbb R^2:
x\in K(m,m';A,u)
\right\}.
\]
Let \(\mathcal F(m,m';A,u)\subseteq J(m,m';A,u)\) denote the
\text{ex-ante} Pareto frontier in payoff space, i.e., the set of Pareto-undominated payoff
pairs in \(J(m,m';A,u)\).

Equivalently, exposed points of the frontier are selected by weighted planner problems of the form
\[
\max_{x\in K(m,m';A,u)}
\lambda\,p_m\cdot x+(1-\lambda)\,p_{m'}\cdot x,
\qquad \text{for } \lambda\in[0,1].
\]
The constraints defining \(K(m,m';A,u)\) impose \text{ex-ante} participation for both agents, while the objective weights the two agents' \text{ex-ante} surpluses.

Our goal is to understand how this frontier changes as agreement increases in the inclusion order. If \((m,m')\succeq^I(\hat m,\hat m')\), then \(C(p_{\hat m},p_{\hat m'})\subseteq C(p_m,p_{m'})\),
so moving up the inclusion order from \((\hat m,\hat m')\) to \((m,m')\) enlarges the set of surplus vectors that satisfy both participation constraints.

A first instinct is to conclude that the planner must be weakly better off under the more agreeable pair, since the jointly acceptable set expands. This conclusion would be correct if
beliefs only affected the \textit{constraints}. The difficulty is that beliefs also enter the \textit{objective}: the same surplus vector \(x\) is evaluated using different inner products under
\((p_m,p_{m'})\) and under \((p_{\hat m},p_{\hat m'})\). Hence, moving up in the inclusion order simultaneously enlarges the set of jointly acceptable surplus vectors and reweights how any given surplus vector maps into agents' \text{ex-ante} payoffs.

The inclusion order disciplines this reweighting. By \Cref{prop:inclusion_char}, the cone inclusion
\(C(p_{\hat m},p_{\hat m'})\subseteq C(p_m,p_{m'})\) holds if and only if there exist \(\alpha,\alpha'\in[0,1]\) such that
\(p_m=\alpha p_{\hat m}+(1-\alpha)p_{\hat m'}\) and \(p_{m'}=\alpha' p_{\hat m}+(1-\alpha')p_{\hat m'}\). Equivalently, there exists a row-stochastic matrix
\[
M
\equiv
\begin{pmatrix}
\alpha & 1-\alpha\\
\alpha' & 1-\alpha'
\end{pmatrix}
\]
such that
\(M(p_{\hat m},p_{\hat m'})^\top=(p_m,p_{m'})^\top\). Consequently, for every surplus vector \(x\), \((p_m\cdot x, p_{m'}\cdot x) = M(p_{\hat m}\cdot x, p_{\hat m'}\cdot x)\), i.e., the same matrix \(M\) records how the more agreeable pair mixes the less agreeable
pair and how payoff evaluations are transformed.

Our next result says that the Pareto frontier under the more agreeable pair (weakly) dominates the \(M\)-transformed payoff set generated by the less agreeable pair. For subsets \(Y,Z\subseteq\mathbb R^2\), define the \textit{weak set order} \(Y\succeq_{WSO} Z\) if, for every \(z\in Z\), there exists \(y\in Y\) such that \(y\ge z\) coordinatewise.

\begin{proposition}\label{prop:pareto_frontier}
If \(C(p_{\hat m},p_{\hat m'})\subseteq C(p_m,p_{m'})\), then there exists a row-stochastic matrix \(M\) with \(M(p_{\hat m},p_{\hat m'})^\top=(p_m,p_{m'})^\top\) such that, for every decision problem \((A,u)\) and reservation payoff \(\bar u\),
\(\mathcal F(m,m';A,u)\ \succeq_{WSO}\ M\left[J(\hat m,\hat m';A,u)\right]\).
\end{proposition}

In words, more agreement, as measured via the inclusion order, forces $(p_m,p_{m'})$ to be a mixture of $(p_{\hat m},p_{\hat m'})$, and the same mixing provides a decision-problem-uniform benchmark for how the frontier moves: the new Pareto frontier coordinatewise dominates the mixed image of the old payoff set. More agreement must lead to an expansion of the appropriately constructed ``normalized'' Pareto frontier.  

Our discussion of Pareto optimal plans mirrors the discussion of admissibility for single decision makers.  In Wald's statistical decision theory \citep{wald1950statistical}, decision rules are compared by their risk vectors, admissibility is undominatedness under coordinatewise risk dominance, and the complete class theorem says (under standard regularity) that admissible rules are Bayes (or limits of Bayes), i.e., minimizers of a prior-weighted linear functional of risk. Here, for the fixed decision problem \((A,u)\) and reservation payoff \(\bar u\), each feasible protocol \(\sigma^{(A,u)}\) induces the \text{ex-ante} payoff vector \(\left(p_m\cdot x_\sigma,p_{m'}\cdot x_\sigma\right)\in J(m,m';A,u)\), so admissibility is again coordinatewise undominatedness. Since \(J(m,m';A,u)\) is convex, every admissible payoff vector is supported by some welfare weights \(\gamma\coloneqq\left(\gamma_1,\gamma_2\right)\in\mathbb{R}^2_+\setminus\left\{0\right\}\); equivalently, \(x_\sigma\) solves
\[x_\sigma\in \arg\max_{x\in X(A,u)} \gamma_1 p_m\cdot x+\gamma_2 p_{m'}\cdot x \quad\text{s.t.}\quad p_m\cdot x\ge0,\ p_{m'}\cdot x\ge0.\]
Thus, varying \(\gamma\) plays the role of varying priors in Wald: welfare-weighted maximizers form a complete class for the Pareto-dominance order on \text{ex-ante} feasible protocols.

We can again illustrate the result using the experiments from Example~\ref{ex:basic}. Consider the three-action decision problem
\begin{align*}
  a_1 &: \quad u(a_1,\omega_1)=3,\quad u(a_1,\omega_2)=-4,\\
  a_2 &: \quad u(a_2,\omega_1)=1,\quad u(a_2,\omega_2)=-1,\\
  a_3 &: \quad u(a_3,\omega)=0.
\end{align*}
Set the reservation payoff to \(\bar u=0\). The action \(a_3\) represents this benchmark, since \(\sum_{\omega\in\Omega}p(\omega)u(a_3,\omega)=0\). After \(s_2\), both bets have weakly negative value under every experiment in the example, so Pareto-undominated acceptable protocols choose \(a_3\) after \(s_2\). The only relevant choice is, therefore, at \(s_1\). 

Let
\(r_1=\mathbb{P}(a_1\mid s_1)\) and \(q_1=\mathbb{P}(a_2\mid s_1)\), with \(r_1+q_1\leq 1\). For this family of protocols, write \(G(m,r_1,q_1)\) for the \text{ex-ante} gain under model
\(m\). A direct calculation produces
\[
\begin{aligned}
G(A,r_1,q_1)   &= \frac14 r_1+\frac{3}{20}q_1,
&\qquad
G(A',r_1,q_1) &= -\frac14 r_1,\\
G(B,r_1,q_1)   &= \frac18 r_1+\frac{9}{80}q_1,
&\qquad
G(B',r_1,q_1) &= -\frac18 r_1+\frac{3}{80}q_1,\\
\text{and} \quad G(C,r_1,q_1)   &=G(C',r_1,q_1)=\frac{3}{40}q_1.
\end{aligned}
\]

For \((A,A')\), the participation constraint for \(A'\) forces \(r_1=0\). Consequently, the frontier is
\(\mathcal F(A,A';A,u)
=
\left\{
\left(\frac{3}{20},0\right)
\right\}\),
achieved by recommending \(a_2\) after \(s_1\). For \((B,B')\), the participation constraint for \(B'\) is
\[
-\frac18 r_1+\frac{3}{80}q_1\geq 0,
\qquad\text{or equivalently}\qquad
q_1\geq \frac{10}{3}r_1.
\]
Together with \(r_1+q_1\leq1\), this yields the frontier segment connecting
\[
P_1=\left(\frac{3}{26},0\right)
\qquad\text{and}\qquad
P_2=\left(\frac{9}{80},\frac{3}{80}\right).
\]
The first point is achieved at \(r_1=\frac{3}{13}\), \(q_1=\frac{10}{13}\); the second is
achieved at \(r_1=0\), \(q_1=1\). Moving along the segment shifts the protocol from more
weight on \(a_1\) to more weight on \(a_2\), reducing \(B\)'s gain while increasing \(B'\)'s
gain.

For \((C,C')\), the agents have identical evaluations. Since
\[
G(C,r_1,q_1)=G(C',r_1,q_1)=\frac{3}{40}q_1,
\]
the frontier is
\(\mathcal F(C,C';A,u)
=
\left\{
\left(\frac{3}{40},\frac{3}{40}\right)
\right\}\).

When agents interpret the same signal differently, the choice of a jointly acceptable protocol may redistribute surplus across agents because the agents evaluate the same protocol differently. As agreement increases, these valuation differences shrink in the sense captured by the mixing operator in \Cref{prop:pareto_frontier}: the new frontier need not contain the old frontier in raw payoff space, but it weakly dominates the appropriately transformed image of the old payoff set. This normalization is necessary.

The three pairs in Example~\ref{ex:basic} share a common midpoint, C.  Total
\text{ex-ante} surplus of any protocol \(x\) is
\(p_m\cdot x+p_{m'}\cdot x=(p_m+p_{m'})\cdot x=2\,p_C\cdot x\), and \(p_m+p_{m'}\)
is constant along the three pairs, the total-surplus function is identical for all three
pairs. Concretely, \(G(m,r_1,q_1)+G(m',r_1,q_1)=\tfrac{3}{20}q_1\) for
\((A,A')\), \((B,B')\), and \((C,C')\): the \(a_1\)-terms cancel across agents and the
\(a_2\)-terms coincide. In particular, the maximal joint surplus \(\tfrac{3}{20}\) is
the same for every pair and is attained by the \emph{same} protocol
\(\sigma^\ast\)---recommend \(a_2\) after \(s_1\), i.e.\ \(r_1=0,\ q_1=1\).

Thus the movement of the actual frontier as agreement increases is not driven by changes in the suprlus of a given protocol, but rather by which protocols are Pareto undominated.  The frontier collapses to a point at
both ends of the chain, but for opposite reasons. Under \((A,A')\), agent~2 holds the
uninformative experiment and so their payoff is pinned won at the reservation payoff: \(G(A')=-\tfrac14 r_1\)
forces \(r_1=0\), so agent~1 captures the entire surplus \(\tfrac{3}{20}\). Under \((C,C')\), the agents are identical, so every protocol is
valued identically and the unique optimum lies on the diagonal. Only the intermediate
pair \((B,B')\) yields a nondegenerate frontier, because agent~2 now earns strictly
positive surplus at \(\sigma^\ast\) (\(G(B')=\tfrac{3}{80}>0\)).  This can transferred: shifting weight onto \(a_1\), which agent~1 values and agent~2 opposes, transfers \text{ex-ante} payoff toward agent~1 until agent~2's participation constraint binds at \(P_1\).  
This transfer is ``costly''. Since total surplus equals \(\tfrac{3}{20}q_1\) and the resource
constraint \(r_1+q_1\le1\) trades \(q_1\) for \(r_1\), every unit of payoff moved toward
agent~1 via \(a_1\) destroys surplus by displacing the surplus-generating action \(a_2\).
Total surplus therefore falls from \(\tfrac{3}{20}\) at \(P_2\) to \(\tfrac{3}{26}\) at
\(P_1\) as one moves along the frontier. Such surplus-destroying protocols are still
Pareto-undominated  because Pareto efficiency ranks agents' individual payoffs,
not their sum.  

The raw frontier is therefore non-monotone in the inclusion order, both in ``size''
(point, then segment, then point) and in the total surplus it can support (which drops
below \(\tfrac{3}{20}\) only in the interior case). This does not conflict with
\Cref{prop:pareto_frontier}, which compares the frontier under the more agreeable pair to
the \(M\)-\emph{transformed} image of the less agreeable pair's payoff set, not to its raw
image. The invariance of \(p_m+p_{m'}\) along this chain makes the point sharply: the
efficient \emph{surplus} is fixed, and increasing agreement only contracts the menu of
Pareto-efficient \emph{splits}---non-monotonically in raw payoff space, but monotonically
once normalized by \(M\).


\subsection{Rationalizing Model}\label{sec:rationalizing}

In many settings an outside observer may see both the signal and the action jointly taken by the two decision makers, but not the internal models they use to map signals to beliefs. The agents may then need to defend their signal-contingent protocol as the outcome of a coherent common narrative. For instance, two radiologists may disagree about how to interpret X-ray scans, but may still need to justify a joint diagnostic procedure in court. This motivates the
question: when does greater agreement make it easier to rationalize jointly-acceptable behavior as Bayesian behavior under a single model?

We formalize rationalizability as compatibility with Bayesian behavior. Take a full-support prior \(p\). Given a rationalizing model \(m^r\), write \(p^r\) for the joint distribution induced by \(m^r\) and \(p\): \(p^r(\omega,s)\coloneqq p(\omega)m^r(s\mid\omega)\). For such a model \(m^r\), say that a protocol \(\sigma^\ast_{(A,u)}\colon S\to\Delta(A)\) is \textit{Bayes-optimal} under \(m^r\) if, for each signal \(s\in S\),
\[
\sigma^\ast_{(A,u)}(\cdot\mid s)
\in
\argmax_{\alpha\in\Delta(A)}
\sum_{\omega\in\Omega}p^r(\omega\mid s)U(\alpha,\omega).
\]
Equivalently, since expected utility is linear in mixed actions,
\[
\supp\left(\sigma^\ast_{(A,u)}(\cdot\mid s)\right)
\subseteq
\argmax_{a\in A}
\sum_{\omega\in\Omega}p^r(\omega\mid s)u(a,\omega).
\]
Accordingly, the protocol may randomize only among pure actions that are Bayes-optimal under the posterior induced by \(m^r\).

Because rationalization requires optimality with respect to feasible actions, we impose a mild feasibility requirement on the reservation benchmark in this subsection. Given a decision problem \((A,u)\) and prior \(p\), say that a reservation payoff \(\bar u\) is \textit{feasibly represented} if there exists an action \(a^0\in A\) such that
\[
\sum_{\omega\in\Omega}p(\omega)u(a^0,\omega)=\bar u.
\]
The representing action need not deliver \(\bar u\) in every state and need not be prior-optimal.

In turn, a model \(m^r\) \textit{rationalizes} a pair of models \((m,m')\) at prior \(p\) if its induced joint distribution \(p^r\) satisfies \(p^r(s)>0\) for every \(s\in S\) and, for every decision problem \((A,u)\) and every feasibly represented reservation payoff \(\bar u\), there exists a Bayes-optimal protocol under \(m^r\) whose induced surplus vector is jointly acceptable for \((m,m')\), i.e.,
\[
x_{\sigma^\ast}\in C(p_m,p_{m'}).
\]
We denote the set of models rationalizing \((m,m')\) at prior \(p\) by \(\mathcal{R}_p(m,m')\). A rationalizing model is, therefore, a common Bayesian narrative under which agents can always find an optimal protocol that is also an \textit{ex-ante} improvement for both agents relative to any feasible reservation benchmark. For the analogous prior-free comparison, write
\[
\mathcal{R}(m,m') \coloneqq
\bigcap_{p\text{ full support}}\mathcal{R}_p(m,m').
\]

The key observation is that agreement relaxes only the acceptability requirement, and so once the jointly acceptable cone expands, any rationalizing model that worked before continues to work.
\begin{proposition}\label{prop:incrat_fixedp}\label{cor:incrat_priorfree}
Fix a full-support prior \(p\). If
\((m,m')\succeq^{I}_p(\hat m,\hat m')\),
then \(\mathcal{R}_p(m,m')\supseteq\mathcal{R}_p(\hat m,\hat m')\).
Moreover, if
\((m,m')\succeq^{I}(\hat m,\hat m')\),
then
\(\mathcal{R}(m,m')\supseteq\mathcal{R}(\hat m,\hat m')\).
\end{proposition}

\Cref{prop:incrat_fixedp} provides a distinct interpretation of the inclusion order.\footnote{We provide a direct proof of this result, but it can also be shown that this is an implication of the elegant characterization of \cite{morris1997rationality}.}  Greater agreement expands the cone of jointly acceptable surplus vectors. Since rationalization requires a Bayes-optimal protocol to lie in that cone, greater agreement enlarges the set of (single-model) Bayesian narratives that can rationalize jointly acceptable behavior.

We again illustrate the mechanism using the experiments from Example~\ref{ex:basic} and the decision problem from Section~\ref{sec:pareto}. Focus on protocols that choose \(a_3\) after
\(s_2\), and write \(r_1=\mathbb{P}(a_1\mid s_1)\) and \(q_1=\mathbb{P}(a_2\mid s_1)\). Let \(G(m,r_1,q_1)\) denote the \text{ex-ante} gain under model \(m\), as in the Pareto-frontier
example. Since \(p_C=\frac12 p_A+\frac12 p_{A'}\), we have, for every such protocol,
\[
G(C,r_1,q_1)
=
\frac12 G(A,r_1,q_1)+\frac12 G(A',r_1,q_1).
\]
Thus, every protocol in this family that is jointly acceptable for \((A,A')\) is also acceptable
for \((C,C')\).

The inclusion can be strict. To see this, let $\Omega=\{\omega_1,\omega_2\}$,
$S=\{s_1,s_2\}$, and let the prior be uniform. Consider
\[
A=
\begin{pmatrix}
0.7&0.3\\
0.4&0.6
\end{pmatrix},
\qquad
A'=
\begin{pmatrix}
0.5&0.5\\
0.5&0.5
\end{pmatrix},
\]
and define
\[
C=C'=\frac12 A+\frac12 A'
=
\begin{pmatrix}
0.6&0.4\\
0.45&0.55
\end{pmatrix}.
\]
Then
\[
C\in\mathcal R_p(C,C')
\qquad\text{but}\qquad
C\notin\mathcal R_p(A,A').
\]

First, \(C\in\mathcal R_p(C,C')\). For any decision problem \((A,u)\) and reservation payoff \(\bar u\) represented by a feasible action \(a^0\), choose after each signal a \(C\)-posterior-optimal action. Because \(a^0\) is feasible after each signal, the chosen posterior-optimal action weakly beats \(a^0\) signal by signal under \(C\). Therefore,
\[
\sum_{s\in S} p_C(s)\max_{\alpha\in\Delta(A)}
\sum_{\omega\in\Omega}p_C(\omega\mid s)U(\alpha,\omega)
\ge
\sum_s p_C(s)\sum_{\omega\in\Omega}p_C(\omega\mid s)u(a^0,\omega)
=
\sum_{\omega\in\Omega}p(\omega)u(a^0,\omega)
=
\bar u.
\]
Since \(C=C'\), both participation constraints are satisfied. It remains to show that $C\notin\mathcal R_p(A,A')$. Consider a decision problem with two actions and reservation payoff \(\bar u=0\). Let \(a_0\) represent the reservation payoff, with zero payoff in both states, while action \(a_1\) yields
\[
u(a_1,\omega_1)=1,
\qquad
u(a_1,\omega_2)=-\frac65 .
\]
The prior expected payoff of \(a_1\) is
\[
\frac12-\frac12\cdot\frac65=-\frac1{10}<0.
\]

Under $C$, after signal $s_1$,
\[
p_C(\omega_1\mid s_1)
=
\frac{0.6}{0.6+0.45}
=
\frac47,
\]
and the posterior payoff of $a_1$ is
\[
\frac47-\frac65\cdot\frac37
=
\frac2{35}>0.
\]
After signal $s_2$,
\[
p_C(\omega_1\mid s_2)
=
\frac{0.4}{0.4+0.55}
=
\frac8{19},
\]
and the posterior payoff of $a_1$ is
\[
\frac8{19}-\frac65\cdot\frac{11}{19}
=
\frac{40-66}{95}<0.
\]
Thus the unique $C$-Bayes-optimal protocol chooses $a_1$ after $s_1$ and $a_0$
after $s_2$.

Under $A'$, signals are uninformative and $p_{A'}(s_1)=1/2$. Conditional on
$s_1$, the expected payoff of $a_1$ remains $-1/10$. Hence the \textit{ex-ante} surplus, measured relative to \(\bar u=0\), of the unique \(C\)-Bayes-optimal protocol under \(A'\) is
\[
\frac12\left(-\frac1{10}\right)
=
-\frac1{20}<0.
\]
Therefore this $C$-Bayes-optimal protocol violates $A'$'s participation
constraint, so $C\notin\mathcal R_p(A,A')$. Hence the inclusion is strict.

The takeaway of this application is that agreement expands the set of externally defensible common narratives. A rationalizing model must do two things at once: it must make the observed protocol
Bayes-optimal, and the protocol it rationalizes must be jointly acceptable to the disagreeing agents. Greater agreement relaxes the second requirement by enlarging the cone of jointly acceptable surplus vectors. Thus, the inclusion order has not only a feasibility interpretation, but also an external-rationalizability interpretation: more agreement makes joint behavior easier to defend as if it were generated by a common model.

\section{Completing the Inclusion Preorder: Cosine Similarity}\label{sec:cosine}

The previous section highlighted the economic content of the inclusion preorder. Its decision-theoretic force, however, comes at the cost of incompleteness. At a fixed full-support prior \(p\), one pair ranks above another exactly when it makes a weakly larger
set of surplus vectors jointly acceptable:
\[
(m,m')\succeq^I_p(\hat m,\hat m')
\quad\Longleftrightarrow\quad
C(p_{\hat m},p_{\hat m'})\subseteq C(p_m,p_{m'}).
\]
When these cones are not nested, the inclusion preorder is silent. This section asks how to complete the comparison in a disciplined way.

Write \(d\coloneqq|\Omega||S|\) and identify \(\mathbb{R}^{\Omega\times S}\) with \(\mathbb{R}^d\). We impose a neutrality requirement on this Arrow-Debreu surplus space: because the comparison ranges over all surplus vectors, it should not depend on the orthonormal coordinates used to represent that space. Orthogonal transformations preserve lengths, angles, and inner products among surplus vectors, and should be interpreted as changes of basis rather than as literal relabelings of states and signals.

This neutrality desideratum leads naturally to an angle-based comparison. The cone \(C(p_m,p_{m'})\) is homogeneous in surplus vectors and unchanged by positive rescalings of either normal vector:
\(C(\alpha p_m,\beta p_{m'})=C(p_m,p_{m'})\) for all \(\alpha,\beta>0\). \textit{Viz.}, the cone depends on \(p_m\) and \(p_{m'}\) only through their directions. The natural rotation-invariant object is, therefore, the angle between these directions.

For nonzero \(a,b\in\mathbb{R}^d\), define
\[
\cos(a,b)\coloneqq\frac{a\cdot b}{||a|| \ ||b||}.
\]
In our application, \(p_m\) and \(p_{m'}\) lie in the positive orthant, so
\(\cos(p_m,p_{m'})\in[0,1]\). Let \(\mathcal{P}_p\subseteq\mathbb{R}^d\) denote the set of
joint-distribution vectors \(p_m\) induced by experiments \(m\in\mathcal{M}\) under prior
\(p\).

\begin{defn}\label{def:ric}
Fix a full-support prior \(p\). A complete preorder \(\succeq\) on \(\mathcal{P}_p^2\) is a rotation-invariant strict completion of the inclusion preorder if it satisfies:
\begin{enumerate}
\item \textit{Extension of inclusion}: if \(C\left(p_{\hat m},p_{\hat m'}\right)\subseteq C\left(p_m,p_{m'}\right)\), then \((p_m,p_{m'})\succeq(p_{\hat m},p_{\hat m'})\).
\item \textit{Rotation invariance}: if there exists \(Q\in O(d)\) such that \(C\left(p_{\hat m},p_{\hat m'}\right)=Q\left[C\left(p_m,p_{m'}\right)\right]\), then \((p_m,p_{m'})\sim(p_{\hat m},p_{\hat m'})\).
\item \textit{Strictness}: if there exists \(Q\in O(d)\) such that \(Q\left[C\left(p_{\hat m},p_{\hat m'}\right)\right]\subsetneq C\left(p_m,p_{m'}\right)\), then \((p_m,p_{m'})\succ(p_{\hat m},p_{\hat m'})\).
\end{enumerate}
\end{defn}

Our next proposition shows that the requirements of \Cref{def:ric} pin down a unique complete extension of the inclusion order. In particular, once a complete preorder extends cone inclusion, treats rotated cones as equivalent, and ranks strict rotated inclusions strictly, it must rank pairs by cosine similarity.

\begin{proposition}\label{prop:cosine}
Fix a full-support prior \(p\). There exists a unique rotation-invariant strict completion of the inclusion preorder, denoted \(\succeq_p^{RIC}\). Moreover, for all \((p_m,p_{m'}),(p_{\hat m},p_{\hat m'})\in\mathcal{P}_p^2\),
\[
(p_m,p_{m'})\succeq_p^{RIC}(p_{\hat m},p_{\hat m'})
\Longleftrightarrow
\cos\left(p_m,p_{m'}\right)\geq\cos\left(p_{\hat m},p_{\hat m'}\right).
\]
\end{proposition}
We, henceforth, refer to \(\succeq_p^{RIC}\) as the (fixed-prior) cosine order.

The cosine order has a direct interpretation in terms of how the two models value the same surplus vector. Model \(m\) assigns \(x\) the \textit{ex-ante} value \(p_m\cdot x\), while model \(m'\) assigns it \(p_{m'}\cdot x\). Since positive rescalings of \(p_m\) and \(p_{m'}\) leave \(C(p_m,p_{m'})\) unchanged, the relevant evaluations are normalized. Let
\(\tilde p_m\coloneqq\frac{p_m}{||p_m||}\) and \(\tilde p_{m'}\coloneqq\frac{p_{m'}}{||p_{m'}||}\).
Then
\[
\sup_{||x||\leq1}
\left|
\tilde p_m\cdot x-\tilde p_{m'}\cdot x
\right|
=
||\tilde p_m-\tilde p_{m'}||,
\]
and so
\[
\cos(p_m,p_{m'})
=
1-\frac{1}{2}
\left[
\sup_{||x||\leq1}
\left|
\tilde p_m\cdot x-\tilde p_{m'}\cdot x
\right|
\right]^2.
\]
We see that cosine similarity is one minus one half of the squared largest possible difference between the agents' normalized \textit{ex-ante} evaluations over the unit surplus ball. High cosine means that the two normalized evaluations remain close for every normalized surplus direction.

This also connects cosine to speculative trade. In Section~\ref{sec:spec}, the radius of the
transfer interval \(\mathcal{T}_p(x;m,m')\) was
\(\frac12 |p_m\cdot x-p_{m'}\cdot x|\). After normalizing the two valuation functionals, the largest radius over the unit surplus ball is
\[
\sup_{||x||\leq1}
\frac12
\left|
\tilde p_m\cdot x-\tilde p_{m'}\cdot x
\right|
=
\frac12 ||\tilde p_m-\tilde p_{m'}||
=
\sqrt{\frac{1-\cos(p_m,p_{m'})}{2}}.
\]
Therefore, at a fixed prior, ranking pairs by cosine similarity is equivalent to ranking them by the maximal normalized radius of their speculative-transfer intervals, with more agreeable pairs having smaller maximal radii. Inclusion requires pointwise nesting of raw transfer intervals for every surplus vector, whereas cosine compares the worst-case normalized radius.

The geometric intuition behind \Cref{prop:cosine} is simple. Take a pair \((p_m,p_{m'})\) and consider the plane \(\Pi\coloneqq\operatorname{span}\{p_m,p_{m'}\}\subseteq\mathbb{R}^d\).
Only this at-most-two-dimensional subspace matters for \textit{ex-ante} joint acceptability. Furthermore, any surplus vector \(x\) can be decomposed as \(x=x_\Pi+x_\perp\), with
\(x_\Pi\in\Pi\) and \(x_\perp\in\Pi^\perp\). Since
\(p_m\cdot x_\perp=p_{m'}\cdot x_\perp=0\), both participation constraints depend only on \(x_\Pi\).

Thus, after ignoring the null subspace \(\Pi^\perp\), every pair generates a two-dimensional wedge. The opening of this wedge is determined by the angle between the two normal vectors \(p_m\) and \(p_{m'}\): the smaller the angle, the larger the cone of jointly acceptable surplus directions. Rotation invariance lets us align the relevant two-dimensional planes and compare only wedge openings. Since smaller angles are exactly larger cosine similarities, the unique rotation-invariant strict completion ranks pairs by \(\cos(p_m,p_{m'})\). If the two belief vectors are collinear, the wedge degenerates to a single half-space, the limiting
case of perfect agreement.

The unit sphere provides a useful secondary interpretation. Let
\[
\mathbb{S}^{d-1}\coloneqq\{z\in\mathbb{R}^d:||z||=1\},
\]
and let \(\mu\) denote the uniform probability measure on \(\mathbb{S}^{d-1}\).\footnote{Although such a uniform distribution over surpluses is natural, we want to emphasize that this does not necessarily correspond to a uniform probability measure over decision problems, which is why we begin our motivation of the cosine completion through Definition \label{def:ric}.} For a pair
\((m,m')\), denote the share of normalized surplus directions that satisfy both agents' \textit{ex-ante} participation constraints
\[
\pi_p(m,m')
\coloneqq
\mu\left(C(p_m,p_{m'})\cap\mathbb{S}^{d-1}\right).\]
For \(d\geq2\), writing
\(\theta=\arccos(\cos(p_m,p_{m'}))\),
\[
\pi_p(m,m')
=
\frac{\pi-\theta}{2\pi}
=
\frac{1}{4}
+
\frac{1}{2\pi}
\arcsin\left(\cos(p_m,p_{m'})\right),
\]
so \(\pi_p(m,m')\) is strictly increasing in \(\cos(p_m,p_{m'})\).

\begin{proposition}[Uniform-sphere interpretation]\label{prop:urs_cosine}
Fix a full-support prior \(p\). Then,
\[
\pi_p(m,m')\geq\pi_p(\hat m,\hat m')
\Longleftrightarrow
\cos(p_m,p_{m'})\geq\cos(p_{\hat m},p_{\hat m'}).
\]
\end{proposition}

\Cref{prop:urs_cosine} states a probabilistic interpretation of the cosine order. A more agreeable pair is one for which a larger share of directions is jointly acceptable.

Recall that in Example \ref{ex:basic} the pair $D,D'$ was not comparable to any of the other three pairs in the set-inclusion order.  However, they are comparable using the cosine order.  In particular the cosine similarities for $(A,A'), (B,B'), (C,C'), (D,D')$ are respectively $0.9535$, $0.9882$, $1.000$ and $0.9872$, meaning it ranks between $A,A'$ and $B,B'$ (notice also that for the three that are rankable by set-inclusion, the cosine order replicates these).  

The cosine completion is a fixed-prior comparison. One might want, in a way parallel to the inclusion order, to define an extension of the cosine order that is independent of the prior, requiring the cosine comparison to hold for every full-support prior:
\[
(m,m')\succeq^{CS}(\hat m,\hat m')
\Longleftrightarrow
\cos\left(p_m,p_{m'}\right)\geq\cos\left(p_{\hat m},p_{\hat m'}\right)\text{ for every full-support prior }p.
\]
Uniform inclusion implies the corresponding prior-uniform cosine comparison. Indeed, if \((m,m')\succeq^I(\hat m,\hat m')\), then \(C\left(p_{\hat m},p_{\hat m'}\right)\subseteq C\left(p_m,p_{m'}\right)\) for every full-support prior \(p\), so \(\cos\left(p_m,p_{m'}\right)\geq\cos\left(p_{\hat m},p_{\hat m'}\right)\) for every such \(p\). The converse fails. Moreover, unlike the fixed-prior cosine order, the prior-uniform cosine order is generally incomplete.
\begin{proposition}\label{prop:cosinprior}
The prior-uniform cosine order \(\succeq^{CS}\) is incomplete, and \(\succeq^{CS}\neq\succeq^I\).
\end{proposition}

Thus, cosine similarity completes the inclusion order only at a fixed prior. Once cosine comparisons are required to hold uniformly over priors, the order again becomes incomplete and, furthermore, it no longer coincides with the prior-free inclusion order.

\section{Discussion}\label{sec:disc}

In this section, we discuss three features of our (dis)agreement orders. First, we compare the inclusion order to the Blackwell order; second, we relate our cosine similarity order to canonical notions of divergences between probability distributions; third, we discuss an ex-interim formulation of our orders and relate it to the \text{ex-ante} formulation.

\subsection{Blackwell Dominance and Agreement}\label{sec:blackwell}

The best-known partial order on experiments is \emph{Blackwell dominance} \((\succeq^B)\). If \(m\succeq^B m'\), then for every decision problem and every prior, an agent can achieve weakly higher expected utility after observing \(m\) than after observing \(m'\). Like our inclusion preorder, Blackwell dominance is generally incomplete.  Of course, Blackwell dominance compares two experiments, while our ordering compares two pairs of experiments.  However, one can easily ``extend'' Blackwell, and ask what happens if one pair $m,m'$ is more Blackwell extreme than $\hat{m}, \hat{m}'$ --- in other words, all four are comparable via the Blackwell order, with the latter being nested inside the former in that order. 

Blackwell dominance and our agreement orders capture different features of experiments. Blackwell's asks whether one experiment is more informative for matching actions to states. Our order asks whether two agents who disagree about signal likelihoods evaluate
signal-contingent surplus plans in similar directions. Crucially, Blackwell dominance need not track agreement in our sense: a pair of experiments can be farther apart in the Blackwell order while inducing more aligned \text{ex-ante} participation constraints.

The distinction is transparent in the binary state and binary signal case. Let \(\Omega=\{\omega_1,\omega_2\}\) and \(S=\{s_1,s_2\}\), and represent an experiment by its pair of signal-\(s_1\) likelihoods
\[
(a,b)\equiv
\left(m(s_1\mid\omega_2), m(s_1\mid\omega_1)\right)\in[0,1]^2.
\]
Let
\[
G=
\begin{pmatrix}
q_1 & 1-q_1\\
q_0 & 1-q_0
\end{pmatrix}
\]
be a binary garbling, with rows and columns indexed by \((s_1,s_2)\). If \(m'=mG\), then, writing \(v\equiv q_0\) and \(\lambda\equiv q_1-q_0\),
\[
(a',b')=
\left(q_0+(q_1-q_0)a, q_0+(q_1-q_0)b\right) = (v+\lambda a, v+\lambda b).
\]

The parameter \(\lambda\) scales the state contrast, \(b'-a'=\lambda(b-a)\). Assuming that the garbling preserves the meaning of signals (i.e., $q_1\geq q_0$) $\lambda\in [0,1]$ measures attenuation of the contrast, which is the usual loss of informativeness captured by Blackwell's order.\footnote{In general, $\norm{\lambda}$ captures attenuation while the sign of $\lambda$ captures signal-label reversal; e.g., a deterministic swap has $\lambda=-1$ and is information-preserving up to relabeling.}

The parameter \(v\), however, shifts both coordinates simultaneously. A garbling can, thus, reduce informativeness while also changing the direction in which the experiment moves. Our agreement order is sensitive to this second component: it agrees with Blackwell spread only when garblings weaken a fixed interpretation of the signal, so that the relevant experiments
move along a single affine slice.

We can see this by extending Example~\ref{ex:basic} by adding
two experiments \(E\) and \(E'\):
\[
A=
\begin{pmatrix}
0.7&0.3\\
0.4&0.6
\end{pmatrix},
\qquad
E=
\begin{pmatrix}
0.7&0.3\\
0.5&0.5
\end{pmatrix},
\]
\[
E'=
\begin{pmatrix}
0.4&0.6\\
0.3&0.7
\end{pmatrix},
\qquad \text{and} \qquad
A'=
\begin{pmatrix}
0.5&0.5\\
0.5&0.5
\end{pmatrix}.
\]
Thus, \(A\succ^B E\succ^B E'\succ^B A'\).

The inclusion order produces a different comparison. The convex-hull characterization in \Cref{prop:inclusion_char} tells us that \((E,E')\succeq^I(A,A')\) if and only if \(E,E'\in\conv\{A,A'\}\), and \((A,A')\succeq^I(E,E')\) if and only if \(A,A'\in\conv\{E,E'\}\). Neither holds.\footnote{For example, in the \((a,b)\)-coordinates, \(E\in\conv\{A,A'\}\) requires
\(0.7=0.7\alpha+0.5(1-\alpha)\), so \(\alpha=1\), but then the \(a\)-coordinate would be \(0.4\), not \(0.5\). Conversely,
\(A\in\conv\{E,E'\}\) requires
\(0.7=0.7\alpha+0.4(1-\alpha)\), so \(\alpha=1\), but then the \(a\)-coordinate would be \(0.5\), not \(0.4\).}

Figure~\ref{fig:blackwell-square-example} reveals the geometry. The Blackwell chain \(A\to E\to E'\to A'\) ``bends'' away from the fixed-interpretation segment \(\conv\{A,A'\}\). The intermediate experiments \(E\) and \(E'\) are Blackwell-between the endpoints of the outer pair, but they are not convex-between them.

The cosine order also distinguishes agreement from Blackwell spread. At the uniform prior,
\[
\cos(p_A,p_{A'})=\sqrt{\frac{10}{11}}\approx0.953,
\qquad \text{and} \qquad
\cos(p_E,p_{E'})=\sqrt{\frac{128}{165}}\approx0.881.
\]
So, despite \((A,A')\) having a wider Blackwell spread, it is more aligned in the cosine order than \((E,E')\). Crucially, garblings leading from \(A\) to \(E\) to \(E'\) do not merely add noise; they also shift the signal distributions, changing the direction of
the induced joint distributions in \(\mathbb R^4\).

The lesson is that Blackwell spread and agreement capture different features of experiments. Blackwell dominance is about garblings: whether one experiment can be obtained from another by adding noise. The inclusion order is about convexification: whether each experiment in one pair is a mixture of the experiments in another pair. These notions agree when Blackwell movement weakens a fixed interpretation of the signal. They come apart when the garbling path bends.

Although Blackwell spread and the inclusion order are independent, they are not completely unrelated. A wider Blackwell pair cannot dominate a narrower one in the inclusion order.

\begin{proposition}\label{prop:blackwell_inclusion_constraint}
If \(Z\succ^B W\succ^B W'\succ^B Z'\) strictly, then
\((Z,Z')\not\succeq^I(W,W')\).
\end{proposition}

Thus, Blackwell dominance constrains agreement only negatively: the wider Blackwell pair \(Z,Z'\) cannot dominate the narrower one \(W,W'\) in the inclusion order. Note, however, that the narrower pair \(W,W'\) need not dominate the wider one \(Z,Z'\) either. Blackwell spread and agreement therefore remain distinct, with incomparability as the generic outcome when garbling changes the interpretive direction of the signal.

\begin{figure}[t]
\centering
\begin{tikzpicture}
\begin{axis}[
  width=9cm, height=9cm,
  axis equal image,
  xmin=0, xmax=1, ymin=0, ymax=1,
  xtick={0,0.2,...,1}, ytick={0,0.2,...,1},
  minor tick num=1,
  xlabel={$\theta_{11}$},
  ylabel={$\theta_{22}$},
  xlabel style={font=\small},
  ylabel style={font=\small},
  tick label style={font=\footnotesize},
  grid=both,
  grid style={line width=0.2pt, draw=gray!30},
  major grid style={line width=0.3pt, draw=gray!40},
  axis line style={line width=0.6pt},
  clip=false,
]

\addplot[dashed, gray!60, line width=0.7pt]
  coordinates {(0.70,0.60)(0.50,0.50)};
\node[gray!70, font=\small, anchor=north west]
  at (axis cs:0.555,0.527) {$\conv\{A,A'\}$};

\draw[->, thick, blue!70, line width=0.8pt, shorten >=3pt, shorten <=3pt]
  (axis cs:0.70,0.60) -- (axis cs:0.70,0.50);
\node[blue!70, font=\small, anchor=west]
  at (axis cs:0.712,0.550) {$G_1$};

\draw[->, thick, blue!70, line width=0.8pt, shorten >=3pt, shorten <=3pt]
  (axis cs:0.70,0.50) -- (axis cs:0.40,0.70);
\node[blue!70, font=\small, anchor=south east]
  at (axis cs:0.568,0.618) {$G_2$};

\draw[->, thick, blue!70, line width=0.8pt, shorten >=3pt, shorten <=3pt]
  (axis cs:0.40,0.70) -- (axis cs:0.50,0.50);
\node[blue!70, font=\small, anchor=east]
  at (axis cs:0.388,0.600) {$G_3$};

\addplot[navy,      mark=*, mark size=2pt, only marks] coordinates {(0.70,0.60)};
\addplot[navy!65,   mark=*, mark size=2pt, only marks] coordinates {(0.50,0.50)};
\addplot[red!65!black, mark=*, mark size=2pt, only marks] coordinates {(0.70,0.50)};
\addplot[red!45!black, mark=*, mark size=2pt, only marks] coordinates {(0.40,0.70)};

\node[font=\small\bfseries, navy, anchor=south]
  at (axis cs:0.700,0.600) {$A$};
\node[font=\small\bfseries, navy!65, anchor=north]
  at (axis cs:0.500,0.500) {$A'$};
\node[font=\small\bfseries, red!65!black, anchor=west]
  at (axis cs:0.700,0.500) {$E$};
\node[font=\small\bfseries, red!45!black, anchor=south]
  at (axis cs:0.400,0.700) {$E'$};

\end{axis}
\end{tikzpicture}
\caption{Blackwell and agreement use different geometries. The arrows show the garbling chain \(A\to E\to E'\to A'\), so \((A,A')\) is more spread out than \((E,E')\) in the Blackwell sense. The dashed segment is the fixed-interpretation slice \(\conv\{A,A'\}\). Here the Blackwell chain bends away from this segment: \(E\) and \(E'\) are
Blackwell-between the endpoints of \((A,A')\), but are not convex-between them.}
\label{fig:blackwell-square-example}
\end{figure}
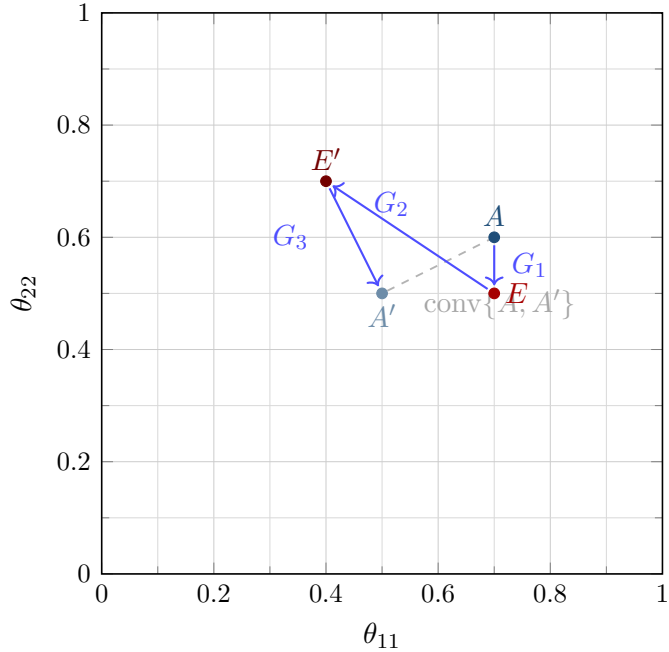

\subsection{Relation to Divergences}\label{sec:divergences}

A large economics literature measures distances between probability distributions using \emph{divergences} --- most prominently the Kullback-Leibler divergence (KL) and its relatives (\cite{kullback1951information}). A useful unifying class is \emph{Bregman divergences}, which are generated by a convex potential and include KL as a special case (\cite{bregman1967relaxation}). This subsection explains how our cosine order fits into that tradition: once we impose the two invariances built into the contracting problem (scale invariance of beliefs and Euclidean relabelings of Arrow-Debreu directions), cosine similarity emerges as the unique Bregman-compatible notion of disagreement.

Let $\phi\colon \mathcal D\to\mathbb R$ be differentiable and strictly convex on an open convex domain $\mathcal D\subseteq\mathbb R^d$.
The associated Bregman divergence is
\[\label{eq:bregdef}\tag{\(1\)}
    D_\phi(u,v)\equiv\phi(u)-\phi(v)-\nabla\phi(v)\cdot (u-v),
\qquad u,v\in\mathcal D.\]
Many familiar divergences between distributions can be written in this form (e.g., KL is the Bregman divergence generated by negative entropy on the simplex). In our setting, however, joint acceptability depends on beliefs only up to positive scaling:
\[
C(\alpha p_m,\beta p_{m'})=C(p_m,p_{m'})
\qquad\text{for all }\alpha,\beta>0.
\]
Accordingly, any divergence-based comparison should be applied after \emph{discarding magnitude} and keeping only direction. This leads us to compare \emph{normalized} belief vectors on the unit sphere:
\[
\tilde p_m \equiv \frac{p_m}{||p_m||},
\qquad
\tilde p_{m'} \equiv \frac{p_{m'}}{||p_{m'}||},
\qquad \text{with }
\tilde p_m,\tilde p_{m'}\in\mathbb S^{d-1}.
\]
Given a potential $\phi$, this suggests a natural divergence preorder:
\[
(m,m')\succeq^\phi_p(\hat m,\hat m')
\quad\Longleftrightarrow\quad
D_\phi(\tilde p_m,\tilde p_{m'})
\leq
D_\phi(\tilde p_{\hat m},\tilde p_{\hat m'}).
\]

The Arrow-Debreu representation endows $\mathbb R^{\Omega\times S}\cong\mathbb R^d$ with a Euclidean structure, and our completion criterion treats orthogonal reparameterizations as irrelevant. This is exactly the content of the rotation invariance axiom in Definition~\ref{def:ric}. Our next observation explains why this axiom is demanding from the perspective of divergences: among Bregman divergences, rotation invariance essentially forces $D_\phi$ to be the Euclidean distance (up to rescaling).
\begin{proposition}[Rotation-invariant Bregman divergences are Euclidean]\label{prop:no-bregman-full}
Let \(d\ge 3\), and let \(\phi\colon \mathcal D\to\mathbb R\) be differentiable and strictly convex on an open set \(\mathcal D\subseteq\mathbb R^d\) containing \(\mathbb S^{d-1}\). Assume the
associated Bregman divergence satisfies orthogonal invariance on the sphere:
\[
D_\phi(Qu,Qv)=D_\phi(u,v),
\qquad
\text{for all }u,v\in\mathbb S^{d-1}\text{ and all }Q\in O(d).
\]
Then there exists \(c>0\) such that, for all \(u,v\in\mathbb S^{d-1}\), \(D_\phi(u,v)=c(1-u\cdot v)=\frac{c}{2}\norm{u-v}^2\). Consequently, for any nonzero \(a,b\in\mathbb R^d\),
\(D_\phi(\tilde a,\tilde b)=c(1-\cos(a,b))\).
\end{proposition}

The proof has two steps. Rotation invariance implies that, on the unit sphere, \(D_\phi(u,v)\) can depend only on the inner product \(u\cdot v\). The Bregman form then forces this dependence to be affine in \(1-u\cdot v\). Thus, the only rotation-invariant Bregman divergence on normalized beliefs is a positive multiple of squared Euclidean distance on the sphere.

The takeaway is that our joint decision problem points to a particular divergence structure. The relevant objects are directions of belief vectors in an Arrow-Debreu surplus space, and the relevant invariance is rotation of that Euclidean space. Once we impose normalization and rotation invariance, the Bregman family collapses to a single comparison. The cosine order is, therefore, not an \textit{ad hoc} similarity index: within the Bregman class, it is exactly the quadratic divergence between normalized belief vectors selected by the same symmetry requirements that underwrite our rotation-invariant completion.

In contrast to our measure, researchers widely use the KL divergence to measure belief distortions, misspecification, and robustness,\footnote{See, for example, work on convergence of beliefs under misspecification \citep{berk1966limiting,esponda_pouzo_2016_berknash,fudenberg2023pathwise}, belief distortions \citep{caplin2019wishful,chambers2023coherent}, and robust decision-making
\citep{hansen2008robustness}.} but it answers a different question from the one addressed by the cosine completion. To see the distinction, consider the symmetrized KL
divergence
\[
\mathrm{KL}_{\mathrm{sym}}(p_m,p_{m'})
=
\frac{1}{2}\left[
\mathrm{KL}(p_m||p_{m'})
+
\mathrm{KL}(p_{m'}||p_m)
\right].
\]
For the three ordered pairs from Example~\ref{ex:basic},
\[
  \mathrm{KL}_{\mathrm{sym}}(A,A') = 0.0525,
  \qquad
  \mathrm{KL}_{\mathrm{sym}}(B,B') = 0.0130,
  \qquad \text{and} \qquad
  \mathrm{KL}_{\mathrm{sym}}(C,C') = 0.
\]
Along this inclusion chain, KL and the inclusion order agree: more agreement corresponds
to smaller symmetrized KL. This is not a coincidence. If
\((m,m')\succeq^I(\hat m,\hat m')\), then \(p_m\) and \(p_{m'}\) lie on the line segment
between \(p_{\hat m}\) and \(p_{\hat m'}\). By joint convexity of KL, the symmetrized KL
between two points on this segment is no larger than the symmetrized KL between the
endpoints. In sum, along inclusion comparisons, KL cannot reverse the direction of the order.

Outside nested pairs, however, KL and cosine can diverge. Consider
\[
  F =
  \begin{pmatrix} 0.5 & 0.5 \\ 0.4 & 0.6 \end{pmatrix},
  \qquad \text{and} \qquad
  F' =
  \begin{pmatrix} 0.6 & 0.4 \\ 0.6 & 0.4 \end{pmatrix}.
\]
These satisfy
\[
  \mathrm{KL}_{\mathrm{sym}}(F,F')
  =
  0.0507
  <
  \mathrm{KL}_{\mathrm{sym}}(A,A')
  =
  0.0525,
\]
but
\[\cos(p_F,p_{F'}) = 0.9515 < \cos(p_A,p_{A'}) = 0.9535.\]
Thus \((F,F')\) is closer than \((A,A')\) in symmetrized KL, but farther apart in the cosine measure. The reversal reflects the different objects being measured. KL compares entry-by-entry log-likelihood ratios. Cosine compares the directions of the induced joint
vectors after discarding scale; and neither statistic dominates the other as a general measure of belief disagreement. The point of \Cref{prop:no-bregman-full} is that, for the decision-theoretic geometry of this paper, cosine is the Bregman comparison selected by
normalization and rotation invariance.

\subsection{\text{Ex-ante} vs. \text{Ex-interim} Individual Rationality}\label{sec:interim}

Our baseline analysis imposes \text{ex-ante} individual rationality: before observing the common signal, each agent must weakly prefer committing to the protocol over receiving the reservation payoff \(\bar u\). In some applications, however, agents may be able to \textit{walk away after observing the signal}. For instance, a committee member may veto a recommendation once the relevant evidence is on the table, or a party to an agreement may renegotiate after receiving a verifiable report. This motivates an \text{ex-interim} version of the contracting problem, in which participation is required \textit{signal by signal}. For the interim formulation, we allow the reservation benchmark to depend on the realized signal and write it as \((\bar u_s)_{s\in S}\).

\medskip
\noindent\textbf{\text{Ex-interim} participation.}
Fix a decision problem \((A,u)\) and signal-contingent interim reservation payoffs \((\bar u_s)_{s\in S}\). For a protocol \(\sigma\), define \(x=(x_s)_{s\in S}\), where \(x_s\in\mathbb R^\Omega\) is the state-contingent interim surplus after signal \(s\),
\[
x_s(\omega)\equiv U(\sigma(\cdot\mid s),\omega)-\bar u_s.
\]
For model \(m\), \text{ex-interim} participation at signal \(s\) requires
\[
p_m(\cdot\mid s)\cdot x_s\geq 0
\qquad\text{whenever }p_m(s)>0,
\]
with the convention that the constraint is vacuous if \(p_m(s)=0\).

For a perceived pair \((m,m')\), interim joint acceptability at signal \(s\) is the intersection of the relevant signal-level half-spaces. When both signal probabilities are
positive, this is
\[
C_s\left(p_m(\cdot\mid s),p_{m'}(\cdot\mid s)\right)
\equiv
\left\{
x_s\in\mathbb R^\Omega\colon 
p_m(\cdot\mid s)\cdot x_s\geq0,\ 
p_{m'}(\cdot\mid s)\cdot x_s\geq0
\right\}.
\]
Vis-a-vis \text{ex-ante} participation, \text{ex-interim} replaces the single joint-acceptability cone
\(C(p_m,p_{m'})\) with a collection of cones in \(\mathbb R^\Omega\),
\(\left(
C_s\left(p_m(\cdot\mid s),p_{m'}(\cdot\mid s)\right)
\right)_{s\in S}\),
one per signal.

\medskip
\noindent\textbf{Link to the \text{ex-ante} dot product.}
The interim viewpoint also clarifies what the analogous \text{ex-ante} dot product averages over. For a surplus vector \(x=(x_s)_{s\in S}\in\mathbb R^{\Omega\times S}\), where \(x_s(\omega)=U(\sigma(\cdot\mid s),\omega)-\bar u_s\), Bayes' rule gives
\[
p_m\cdot x
=
\sum_{s\in S}\sum_{\omega\in\Omega}p_m(\omega,s)x_s(\omega)
=
\sum_{s\in S}p_m(s)\left(p_m(\cdot\mid s)\cdot x_s\right).
\]
Thus, relative to the corresponding signal-contingent benchmark \((\bar u_s)_{s\in S}\), interim participation is stronger than requiring only the signal-probability-weighted average of interim surpluses to be nonnegative: a protocol can be acceptable on average even if it is unacceptable after some rare signal, as long as gains at other signals compensate under the agent's subjective weights.

\medskip
\noindent\textbf{\text{Ex-ante} and \text{ex-interim} agreement can diverge.}
Because Bayes' rule divides by \(p_m(s)\), conditioning can amplify disagreement concentrated on low-probability signals. Conversely, conditioning can erase disagreement that is purely about how frequently signals occur. The following two examples illustrate
the distinction.

First, \text{ex-ante} closeness need not imply interim closeness. Let
\(\Omega=\{0,1\}\), \(S=\{0,1\}\), and \(p(0)=p(1)=1/2\). For \(\varepsilon>0\), consider two
experiments under which signal \(1\) is rare for both agents but has opposite meanings:
\(m(1\mid 1)=\varepsilon\), \(m(1\mid 0)=\varepsilon^2\), \(m'(1\mid 1)=\varepsilon^2\), and \(m'(1\mid 0)=\varepsilon\). Then,
\(p_m(1)=p_{m'}(1)=\frac12(\varepsilon+\varepsilon^2)\), so the joint vectors \(p_m\) and \(p_{m'}\) become close in
\(\mathbb R^{\Omega\times S}\) as \(\varepsilon\to0\). Yet the posteriors at the rare signal
diverge:
\[
p_m(1\mid s=1)
=
\frac{\varepsilon}{\varepsilon+\varepsilon^2}
\to1,
\qquad \text{but} \qquad
p_{m'}(1\mid s=1)
=
\frac{\varepsilon^2}{\varepsilon+\varepsilon^2}
\to0.
\]
We see that the pair can appear close in the \text{ex-ante} geometry while inducing extreme interim disagreement after a rare signal.

Second, interim closeness need not imply \text{ex-ante} closeness. Fix the prior \(p\), and call an experiment \(m\) uninformative if
\(m(s\mid\omega)=\rho(s)\) for some \(\rho\in\Delta(S)\) independent of \(\omega\). Any two uninformative experiments induce identical posteriors signal by signal (the prior), so
their interim cones coincide. But if two such experiments have different signal distributions
\(\rho\) and \(\rho'\), then their joint vectors are
\(p_m(\omega,s)=p(\omega)\rho(s)\) and \(p_{m'}(\omega,s)=p(\omega)\rho'(s)\), so \textit{ex-ante} geometry \textit{does} distinguish them: if \(\rho\) and \(\rho'\) put mass on nearly disjoint subsets of \(S\), then \(p_m\) and \(p_{m'}\) can be far apart in \(\mathbb R^{\Omega\times S}\), despite perfect interim agreement.

The observations above are not pathologies; they reflect a basic feature of Bayesian updating. \text{Ex-ante} participation constraints ``average'' interim incentives using subjective signal probabilities, whereas interim participation conditions strip away those weights and look only at posteriors. Accordingly, \text{ex-ante} and \textit{ex-interim} joint acceptability constraints generate distinct geometric objects and different comparative statics: agreement in the \text{ex-ante} cone need not translate into agreement signal by signal, and vice versa.

This distinction also clarifies the dimensions along which agents may disagree. Following \citet{bohren2023behavioral}, disagreement can be prospective, retrospective, or both. Prospective disagreement concerns the probabilities of signals. Retrospective disagreement
concerns the meaning of realized signals, as captured by posterior beliefs. The \text{ex-ante} geometry is sensitive to both components because it depends on the full joint distributions \(p_m(\omega,s)\) and \(p_{m'}(\omega,s)\). The \text{ex-interim} geometry
isolates the retrospective component because it compares the posteriors \(p_m(\cdot\mid s)\) and \(p_{m'}(\cdot\mid s)\) signal by signal.\footnote{If we specify an objective signal-generating process, we can sharpen the taxonomy further by requiring subjective models to match the true unconditional distribution of signals, the
property \citet{bohren2023behavioral} call introspection-proofness. We leave this refinement to Online Appendix \ref{app:introspection}.}

The takeaway is that \text{ex-ante} and ex-interim participation capture distinct notions of agreement. The baseline \text{ex-ante} benchmark asks whether agents are willing to commit before the signal is observed, and, therefore, aggregates gains and losses relative to the scalar reservation payoff \(\bar u\) across signals using each agent's subjective signal probabilities. The ex-interim benchmark asks whether agents remain willing to participate after each signal realization relative to the signal-contingent reservation payoff \(\bar u_s\), and, therefore, focuses on posterior disagreement signal by signal. Neither notion subsumes the other as a measure of closeness between subjective models. Rather, the appropriate notion depends on the timing of commitment: \text{ex-ante} agreement is the relevant object when agents can commit before information arrives, while ex-interim agreement becomes relevant when agents can veto, renegotiate, or walk away after observing the signal.

\section{Conclusion}\label{sec:con}

This paper develops a decision-theoretic measure of disagreement between subjective signal structures. The primitive object is a pair of models for interpreting a common signal. Rather than measuring statistical distance directly, we ask which signal-contingent surplus plans both agents are willing to accept \text{ex-ante}, relative to the problem's reservation payoff. This yields a preorder over pairs of subjective models: one pair is more agreeable than another if it supports a larger cone of jointly acceptable surplus vectors.

Our main characterization reveals that this behavioral comparison has a simple representation. One pair of experiments is more agreeable than another if and only if each model in the more agreeable pair is a convex combination of the two models in the less agreeable pair. In short, a criterion defined through \text{ex-ante} participation constraints reduces to a prior-independent convexity test on signal structures themselves. The resulting order has direct economic content: greater agreement shrinks speculative-trade wedges, expands appropriately normalized \text{ex-ante} Pareto frontiers, and enlarges the set of single-model Bayesian narratives that can rationalize jointly acceptable behavior.

The inclusion preorder is incomplete, but its geometry also identifies a canonical scalar completion. Requiring a complete comparison to extend cone inclusion and to be invariant to rotations of the Arrow--Debreu surplus space uniquely selects cosine similarity between the induced joint state-signal distributions. Equivalently, within the class of rotation-invariant Bregman comparisons on normalized belief vectors, the relevant divergence is quadratic rather
than of KL-type. This distinction reflects the object being measured: not statistical proximity \textit{per se}, but similarity in the agents' evaluations of signal-contingent surplus directions.

Several extensions are natural. If agents disagree about priors, our fixed-reservation-payoff formulation remains well defined; one could also require participation relative to a set of possible reservation payoffs. If agents disagree about payoffs, the common surplus space itself changes, so the cone geometry would need to be modified. The framework also suggests a robust-choice interpretation, in which a single decision maker is uncertain which signal structure governs the data and accepts a plan only if it improves on the reservation payoff under all candidate experiments. Finally, extending the analysis beyond two agents would replace the intersection of two half-spaces with a higher-dimensional cone generated by many participation constraints.

Our broader message is that disagreement about interpretation can be compared through
its implications for common action. In environments such as medicine, financial analysis,
monetary policy, and algorithmic decision-making, agents often observe the same evidence but
disagree about what it means. The order developed here measures such disagreement by asking
which joint plans remain acceptable to all parties. In that sense, the distance between
interpretations is disciplined by the geometry of joint decision-making.

\newpage

\bibliography{sample.bib}
\bibliographystyle{plainnat}

\appendix

\newpage

    \section{Proofs}\label{sec:proofs}

\subsection{\texorpdfstring{Proof of Lemma \ref{lem:extreme}}{.}}
\begin{proof}
    If \(p_1=p_2\), then the two participation constraints coincide, so \(C(p_1,p_2) = \left\{x \colon p_1\cdot x\ge 0\right\}\), which is a half-space. If \(p_1\cdot p_2=0\), then the two normals are orthogonal. Since \(p_1,p_2\ge 0\), their dot product is always nonnegative, so the angle between them is always in \(\left[0,\pi/2\right]\). Hence, orthogonality is the largest possible angle. Finally, for nonnegative vectors, \(p_1\cdot p_2=0\) is equivalent to disjoint supports.
\end{proof}

\subsection{\texorpdfstring{Proof of \Cref{prop:inclusion_char}}{.}} 

\begin{proof}
    Note that, by definition of $\succeq^{I}_p$ and of dual cones,
\[(m,m')\succeq^{I}_p(\hat{m},\hat{m}')\Leftrightarrow C(p_{m},p_{m'})\supseteq C(p_{\hat{m}},p_{\hat{m}'})\Leftrightarrow (C(p_{\hat{m}},p_{\hat{m}'}))^*\supseteq (C(p_{m},p_{m'}))^*.\]
Recall that $(C(p,p'))^\ast=cone(p,p')\equiv\{x\in\mathbb{R}^{nk}\mid x=\alpha p+\beta p', \text{ for } \alpha,\beta\geq 0\}$. Then, 
\[(m,m')\succeq^{I}_p(\hat{m},\hat{m}')\Leftrightarrow \mathrm{cone}(p_{\hat{m}}, p_{\hat{m}'})\supseteq \mathrm{cone}(p_m, p_{m'})\Leftrightarrow \{p_m, p_{m'}\}\subseteq \mathrm{cone}(p_{\hat m}, p_{\hat m'}).\]
Equivalently, there exist $\alpha,\beta,\alpha',\beta'\geq 0$ such that  
\[\label{eq:a1}\tag{\(A1\)}p_{m}   = \alpha p_{\hat m} + \beta p_{\hat m'}, \qquad \text{and} \qquad
     p_{m'} = \alpha' p_{\hat m} + \beta' p_{\hat m'}.\]
Since $p_m, p_{m'},p_{\hat{m}}, p_{\hat{m}'}$ are probability distributions, we get $\beta=1-\alpha$ and $\beta'=1-\alpha'$.

Finally, \ref{prop1item2} \(\Rightarrow\) \ref{prop1item1} is immediate. If \ref{prop1item1} holds, then the fixed-prior characterization yields \eqref{eq:a1} for some full-support prior \(p\). Since \(p_m(\omega,s)=p(\omega)m(s\mid\omega)\), and similarly for \(m',\hat m,\hat m'\), dividing by \(p(\omega)>0\) delivers \ref{prop1item3}. Conversely, if \ref{prop1item3} holds, then multiplying the two assumed identities by any full-support prior \(p\) produces the fixed-prior characterization at that prior.
Thus, \((m,m')\succeq^I_p(\hat m,\hat m')\) for every full-support prior \(p\), which is \ref{prop1item2}.
\end{proof}

\subsection{\texorpdfstring{Proof of \Cref{prop:spec}}{.}}
\begin{proof}
    Assume first that \((m,m')\succeq^I_p(\hat m,\hat m')\). \Cref{prop:inclusion_char} tells us that there exist \(\alpha,\alpha' \in [0,1]\) such that
\[p_m = \alpha p_{\hat m}+\left(1-\alpha\right)p_{\hat m'},\qquad \text{and} \qquad
p_{m'} = \alpha' p_{\hat m}+\left(1-\alpha'\right)p_{\hat m'}.\]
Taking dot products with \(x\) yields
\[p_m\cdot x = \alpha\left(p_{\hat m}\cdot x\right)+\left(1-\alpha\right)\left(p_{\hat m'}\cdot x\right),\]
and
\[p_{m'}\cdot x = \alpha'\left(p_{\hat m}\cdot x\right)+\left(1-\alpha'\right)\left(p_{\hat m'}\cdot x\right).\]
We conclude that both \(p_m\cdot x\) and \(p_{m'}\cdot x\) lie in \(\conv\left\{p_{\hat m}\cdot x,p_{\hat m'}\cdot x\right\}\). Therefore,
\[\mathcal{T}_p\left(x;m,m'\right) \subseteq \mathcal{T}_p\left(x;\hat m,\hat m'\right).\]

Conversely, suppose that
\[\mathcal{T}_p\left(x;m,m'\right) \subseteq \mathcal{T}_p\left(x;\hat m,\hat m'\right) \quad \forall \ x \in \mathbb{R}^{\Omega\times S}.\]
As \(p_m\cdot x \in \mathcal{T}_p\left(x;m,m'\right)\), we have
\(p_m\cdot x \in \mathcal{T}_p\left(x;\hat m,\hat m'\right)\) for every \(x\). We claim that
\(p_m \in \conv\left\{p_{\hat m},p_{\hat m'}\right\}\). If not, then because \(\conv\left\{p_{\hat m},p_{\hat m'}\right\}\) is compact and convex, the strict separating hyperplane theorem implies that there exists \(x \in \mathbb{R}^{\Omega\times S}\) such that either
\[p_m\cdot x>\sup_{y \in \conv\left\{p_{\hat m},p_{\hat m'}\right\}}y\cdot x \qquad \text{or} \qquad p_m\cdot x<\inf_{y \in\conv\left\{p_{\hat m},p_{\hat m'}\right\}}y\cdot x.\]
Since a linear functional attains its maximum and minimum over a line segment at an endpoint, this implies either
\[p_m\cdot x>\max\left\{p_{\hat m}\cdot x,p_{\hat m'}\cdot x\right\} \qquad \text{or} \qquad
p_m\cdot x<\min\left\{p_{\hat m}\cdot x,p_{\hat m'}\cdot x\right\},\]
contradicting \(p_m \cdot x \in \mathcal{T}_p\left(x;\hat m,\hat m'\right)\). Therefore, \(p_m \in \conv\left\{p_{\hat m},p_{\hat m'}\right\}\) and the same argument delivers \(p_{m'} \in\conv\left\{p_{\hat m},p_{\hat m'}\right\}\).

Hence, there exist \(\alpha,\alpha' \in[0,1]\) such that
\[p_m = \alpha p_{\hat m}+\left(1 - \alpha\right)p_{\hat m'},\qquad \text{and} \qquad
p_{m'} = \alpha' p_{\hat m}+\left(1-\alpha'\right)p_{\hat m'}.\]
By \Cref{prop:inclusion_char}, \((m,m')\succeq^I_p(\hat m,\hat m')\).\end{proof}

\subsection{\texorpdfstring{Proof of \Cref{prop:pareto_frontier}}{.}}
\begin{proof}
    By \Cref{prop:inclusion_char}, the inclusion
\(C(p_{\hat m},p_{\hat m'})\subseteq C(p_m,p_{m'})\) implies that there are
\(\alpha,\alpha'\in[0,1]\) such that
\[
p_m=\alpha p_{\hat m}+(1-\alpha)p_{\hat m'},
\qquad \text{and} \qquad
p_{m'}=\alpha' p_{\hat m}+(1-\alpha')p_{\hat m'}.
\]
Let
\[
M=\begin{pmatrix}
\alpha & 1-\alpha\\
\alpha' & 1-\alpha'
\end{pmatrix}.
\]
Then \(M\) is row-stochastic and
\(M(p_{\hat m},p_{\hat m'})^\top=(p_m,p_{m'})^\top\).

Fix \(y\in J(\hat m,\hat m';A,u)\). By definition, there is some
\(x\in K(\hat m,\hat m';A,u)\) such that
\(y=(p_{\hat m}\cdot x,p_{\hat m'}\cdot x)\). Set \(y'=My\). The definition of
\(M\) yields \(y'=(p_m\cdot x,p_{m'}\cdot x)\). Moreover,
\[
K(\hat m,\hat m';A,u)
=
X(A,u)\cap C(p_{\hat m},p_{\hat m'})
\subseteq
X(A,u)\cap C(p_m,p_{m'})
=
K(m,m';A,u),
\]
so the same surplus vector \(x\) is feasible for \((m,m')\). Hence
\(y'\in J(m,m';A,u)\).

It remains only to move from this feasible payoff \(y'\) to a Pareto-frontier point that dominates it. Define
\[
J^+(y')\equiv J(m,m';A,u)\cap(y'+\mathbb R^2_+).
\]
This set is nonempty, since it contains \(y'\), and compact, since
\(J(m,m';A,u)\) is compact. Choose any \(\lambda\gg0\), and let
\(z\in\arg\max_{\tilde z\in J^+(y')}\lambda\cdot\tilde z\). By construction, \(z\ge y'\). We claim that \(z\in\mathcal F(m,m';A,u)\). If not, some \(\bar z\in J(m,m';A,u)\) satisfies \(\bar z\ge z\) and \(\bar z\neq z\). Since \(z\ge y'\), this \(\bar z\) also belongs to \(J^+(y')\). But \(\lambda\gg0\) then implies \(\lambda\cdot\bar z>\lambda\cdot z\), contradicting the choice of \(z\). Thus, for every \(y\in J(\hat m,\hat m';A,u)\), there exists \(z\in\mathcal F(m,m';A,u)\) such that \(z\ge My\). Equivalently, \(\mathcal F(m,m';A,u)\succeq_{WSO}M[J(\hat m,\hat m';A,u)]\).\end{proof}

\subsection{\texorpdfstring{Proof of \Cref{prop:incrat_fixedp}}{.}}
\begin{proof}
Take a full-support prior \(p\) and suppose
\((m,m')\succeq^{I}_p(\hat m,\hat m')\), i.e., \(C(p_{\hat m},p_{\hat m'})\subseteq C(p_m,p_{m'})\). Let \(m^r\in\mathcal{R}_p(\hat m,\hat m')\). By definition, for every decision problem \((A,u)\) and reservation payoff \(\bar u\) whose reservation payoff has a feasible representation, there exists a Bayes-optimal protocol \(\sigma^\ast_{(A,u)}\) under \(m^r\) whose induced surplus vector satisfies \(x_{\sigma^\ast}\in C(p_{\hat m},p_{\hat m'})\). By cone inclusion, \(x_{\sigma^\ast}\in C(p_m,p_{m'})\). The protocol remains Bayes-optimal under \(m^r\), since the rationalizing model has not
changed. Hence, \(m^r\in\mathcal{R}_p(m,m')\); and since \(m^r\) was arbitrary, \(\mathcal{R}_p(m,m')\supseteq\mathcal{R}_p(\hat m,\hat m')\).

For the prior-free statement, suppose \((m,m')\succeq^{I}(\hat m,\hat m')\) and take \(m^r\in\mathcal{R}(\hat m,\hat m')\). Then for every full-support prior \(p\), \(m^r\in\mathcal{R}_p(\hat m,\hat m')\), and uniform inclusion implies \(\mathcal{R}_p(m,m')\supseteq\mathcal{R}_p(\hat m,\hat m')\).
Consequently, \(m^r\in\mathcal{R}_p(m,m')\) for every full-support prior \(p\), so \(m^r\in\mathcal{R}(m,m')\). Therefore,
\(\mathcal{R}(m,m')\supseteq\mathcal{R}(\hat m,\hat m')\).\end{proof}

\subsection{\texorpdfstring{Proof of \Cref{prop:cosine}}{.}}
\begin{proof}
If \(d=1\), every element of \(\mathcal{P}_p\) has the same direction, so every cone is the same half-space and the result is immediate. Assume \(d\ge2\). For \(a,b\in \mathcal{P}_p\), recall \(\theta(a,b)\coloneqq\arccos\left(\cos(a,b)\right)\). Since \(a,b\neq0\) and \(a,b\ge0\), \(\theta(a,b)\in\left[0,\pi/2\right]\).

We first record the elementary geometric fact used below. For all \((a,b),(\hat a,\hat b)\in \mathcal{P}_p^2\),
\[\label{fact}\tag{FACT}
\exists Q\in O(d)\colon Q\left[C(\hat a,\hat b)\right]\subseteq C(a,b) \qquad \Longleftrightarrow \qquad \theta(a,b)\le\theta(\hat a,\hat b),
\]
with strict inclusion iff \(\theta(a,b)<\theta(\hat a,\hat b)\), and equality up to rotation iff \(\theta(a,b)=\theta(\hat a,\hat b)\). Indeed, \(C(a,b)^*=\operatorname{cone}(a,b)\). Hence, for any \(Q\in O(d)\), \(Q\left[C(\hat a,\hat b)\right]\subseteq C(a,b)\) iff \(\operatorname{cone}(a,b)\subseteq Q\operatorname{cone}(\hat a,\hat b)\). These normal cones are one- or two-ray closed wedges, and the aperture of \(\operatorname{cone}(a,b)\) is \(\theta(a,b)\). Thus, a rotated copy of \(\operatorname{cone}(\hat a,\hat b)\) contains \(\operatorname{cone}(a,b)\) iff its aperture is weakly larger. Strictness and equality follow because rotations preserve aperture and two closed wedges with the same aperture can be nested only if they coincide.

Define \(\succeq_p^{RIC}\) on \(\mathcal{P}_p^2\) by \((a,b)\succeq_p^{RIC}(\hat a,\hat b)\) iff \(\cos(a,b)\ge\cos(\hat a,\hat b)\). This relation is complete. If \(C(\hat a,\hat b)\subseteq C(a,b)\), \eqref{fact} with \(Q=I\) gives \(\theta(a,b)\le\theta(\hat a,\hat b)\), hence, \(\cos(a,b)\ge\cos(\hat a,\hat b)\). Thus, \(\succeq_p^{RIC}\) extends the inclusion preorder. If \(C(\hat a,\hat b)=Q\left[C(a,b)\right]\) for some \(Q\in O(d)\), \eqref{fact} in both directions gives \(\theta(a,b)=\theta(\hat a,\hat b)\), hence, \((a,b)\sim_p^{RIC}(\hat a,\hat b)\). If \(Q\left[C(\hat a,\hat b)\right]\subsetneq C(a,b)\), \eqref{fact} produces \(\theta(a,b)<\theta(\hat a,\hat b)\), hence, \(\cos(a,b)>\cos(\hat a,\hat b)\), so \((a,b)\succ_p^{RIC}(\hat a,\hat b)\). Therefore, \(\succeq_p^{RIC}\) is a rotation-invariant strict completion of the inclusion preorder.

It remains to prove uniqueness. Let \(\succeq\) be any rotation-invariant strict completion of the inclusion preorder on \(\mathcal{P}_p^2\). Fix \((a,b),(\hat a,\hat b)\in \mathcal{P}_p^2\). If \(\cos(a,b)=\cos(\hat a,\hat b)\), then \(\theta(a,b)=\theta(\hat a,\hat b)\), so \eqref{fact} yields some \(Q\in O(d)\) such that \(Q\left[C(\hat a,\hat b)\right]=C(a,b)\). Equivalently, \(C(\hat a,\hat b)=Q^{-1}\left[C(a,b)\right]\), and rotation invariance implies \((a,b)\sim(\hat a,\hat b)\). If \(\cos(a,b)>\cos(\hat a,\hat b)\), then \(\theta(a,b)<\theta(\hat a,\hat b)\), so \eqref{fact} yields some \(Q\in O(d)\) such that \(Q\left[C(\hat a,\hat b)\right]\subsetneq C(a,b)\). Strictness implies \((a,b)\succ(\hat a,\hat b)\). The case \(\cos(a,b)<\cos(\hat a,\hat b)\) is symmetric, so \((a,b)\prec(\hat a,\hat b)\).

Thus, for all \((a,b),(\hat a,\hat b)\in \mathcal{P}_p^2\),
\((a,b)\succeq(\hat a,\hat b)\) iff \(\cos(a,b)\ge\cos(\hat a,\hat b)\), and so every rotation-invariant strict completion coincides with \(\succeq_p^{RIC}\). Taking \(a=p_m\), \(b=p_{m'}\), \(\hat a=p_{\hat m}\), and \(\hat b=p_{\hat m'}\) yields the proposition.
\end{proof}

\subsection{\texorpdfstring{Proof of \Cref{prop:urs_cosine}}{.}} 
\begin{proof}
The case \(d=1\) is trivial, so assume \(d\ge2\). For nonzero \(a,b\in P_p\), recall \(\theta(a,b)\coloneqq\arccos\left(\cos(a,b)\right)\). Since \(a,b\ge0\), \(\theta(a,b)\in\left[0,\pi/2\right]\).

We first compute \(\mu\left(C(a,b)\cap\mathbb S^{d-1}\right)\). By rotation invariance of \(\mu\), rotate coordinates so that \(a\) is a positive multiple of \(e_1\) and \(b\) is a positive multiple of \(\cos\left(\theta(a,b)\right)e_1+\sin\left(\theta(a,b)\right)e_2\). Then \(z\in C(a,b)\cap\mathbb S^{d-1}\) iff
\[
z_1\ge0\quad\text{and}\quad\cos\left(\theta(a,b)\right)z_1+\sin\left(\theta(a,b)\right)z_2\ge0.
\]
These inequalities depend only on the polar angle \(\varphi\) of the projection of \(z\) on \(\operatorname{span}\left\{e_1,e_2\right\}\). By invariance under rotations in this plane, this angle is uniform on \(\left[0,2\pi\right)\), except on the \(\mu\)-null set where the projection is zero. For \(y=\left(\cos\varphi,\sin\varphi\right)\), the two inequalities become \(\cos\varphi\ge0\) and \(\cos\left(\varphi-\theta(a,b)\right)\ge0\), whose intersection has angular length \(\pi-\theta(a,b)\). Hence,
\[
\mu\left(C(a,b)\cap\mathbb S^{d-1}\right)=\frac{\pi-\theta(a,b)}{2\pi}.
\]
Applying this identity to \((a,b)=\left(p_m,p_{m'}\right)\) and \((a,b)=\left(p_{\hat m},p_{\hat m'}\right)\), we obtain
\[
\pi_p(m,m')\ge\pi_p(\hat m,\hat m')\Longleftrightarrow\theta\left(p_m,p_{m'}\right)\le\theta\left(p_{\hat m},p_{\hat m'}\right)\Longleftrightarrow\cos\left(p_m,p_{m'}\right)\ge\cos\left(p_{\hat m},p_{\hat m'}\right),
\]
where the last equivalence uses that cosine is decreasing on \(\left[0,\pi\right]\). This proves the result.
\end{proof}

\subsection{\texorpdfstring{Proof of \Cref{prop:cosinprior}}{.}} 
\begin{proof}
We first show that \(\succeq^{CS}\) is incomplete. Let \(\Omega=\left\{\omega_1,\omega_2\right\}\) and \(S=\left\{s_1,s_2\right\}\), and define
\[
m=
\begin{pmatrix}
9/10&1/10\\
1/2&1/2
\end{pmatrix},
\qquad
m'=
\begin{pmatrix}
1/10&9/10\\
1/2&1/2
\end{pmatrix},
\]
\[
\hat m=
\begin{pmatrix}
1/2&1/2\\
9/10&1/10
\end{pmatrix},
\qquad \text{and} \qquad
\hat m'=
\begin{pmatrix}
1/2&1/2\\
1/10&9/10
\end{pmatrix}.
\]
For \(t\in(0,1)\), let \(p_t(\omega_1)=t\) and \(p_t(\omega_2)=1-t\), and write \(p_{t,x}\) for the joint distribution induced by experiment \(x\) under prior \(p_t\). A direct computation produces
\[
\cos\left(p_{t,m},p_{t,m'}\right)=\frac{9t^2+25(1-t)^2}{41t^2+25(1-t)^2},
\qquad \text{and} \qquad
\cos\left(p_{t,\hat m},p_{t,\hat m'}\right)=\frac{25t^2+9(1-t)^2}{25t^2+41(1-t)^2}.
\]
At \(t=1/4\), the first expression equals \(117/133\) and the second equals \(53/197\). At \(t=3/4\), the inequalities reverse. Hence neither pair dominates the other for every full-support prior, so \(\succeq^{CS}\) is incomplete.

Next we show that \(\succeq^{CS}\ne\succeq^I\). Let
\[
r=
\begin{pmatrix}
9/10&1/10\\
1/10&9/10
\end{pmatrix},
\quad
r'=
\begin{pmatrix}
1/10&9/10\\
9/10&1/10
\end{pmatrix},
\qquad \text{and} \qquad
\bar r=
\begin{pmatrix}
7/10&3/10\\
7/10&3/10
\end{pmatrix}.
\]
For every full-support prior \(p\),
\(\cos\left(p_{\bar r},p_{\bar r}\right)=1>\frac{9}{41}=\cos\left(p_r,p_{r'}\right)\), so \((\bar r,\bar r)\succeq^{CS}(r,r')\). But \((\bar r,\bar r)\not\succeq^I(r,r')\). Indeed, if \((\bar r,\bar r)\succeq^I(r,r')\), then by the convex-hull characterization of \(\succeq^I\), there would exist \(\alpha\in[0,1]\) such that
\(\bar r=\alpha r+(1-\alpha)r'\). The \((\omega_1,s_1)\) entry produces \(\alpha=3/4\), while the \((\omega_2,s_1)\) entry produces \(\alpha=1/4\), a contradiction. Therefore, \((\bar r,\bar r)\succeq^{CS}(r,r')\) but \((\bar r,\bar r)\not\succeq^I(r,r')\), so \(\succeq^{CS}\ne\succeq^I\).
\end{proof}

\subsection{\texorpdfstring{Proof of \Cref{prop:blackwell_inclusion_constraint}}{.}}
\begin{proof}
Suppose, toward a contradiction, that \((Z, Z') \succeq^I (W, W')\), i.e., \(Z = \alpha W + (1-\alpha) W'\) for some \(\alpha \in [0,1]\). As \(W \succ^B W'\), there exists a garbling \(G\) such that \(W' = WG\). Thus,
\[Z = \alpha W + (1-\alpha)WG = W\left[\alpha I + (1-\alpha)G\right].\]
Now let \(H \coloneqq \alpha I + (1-\alpha)G\). As both \(I\) and \(G\) are stochastic matrices, so too is \(H\). Thus, \(Z = WH\) and so \(Z\) is a garbling of \(W\). Therefore, \(W \succeq^B Z\), contradicting \(Z \succ^B W\).\end{proof}

\subsection{\texorpdfstring{Proof of \Cref{prop:no-bregman-full}}{.}}
\begin{proof}
For \(u,v\in\mathbb S^{d-1}\), orthogonal invariance implies that \(D_\phi(u,v)\) depends only on \(u\cdot v\). Indeed, if \(u\cdot v=\tilde u\cdot\tilde v\), there is \(Q\in O(d)\) such that \(Qu=\tilde u\) and \(Qv=\tilde v\). Hence, there exists \(f\colon[-1,1]\to\mathbb R\) such that
\(D_\phi(u,v)=f(u\cdot v)\) for all \(u,v\in\mathbb S^{d-1}\).

Fix orthogonal \(y,z\in\mathbb S^{d-1}\). For every \(x\in\mathbb S^{d-1}\), from the Bregman form,
\[\tag{\(A2\)}\label{eq:a2}
f(x\cdot y)-f(x\cdot z)=D_\phi(x,y)-D_\phi(x,z)=\left(\nabla\phi(z)-\nabla\phi(y)\right)\cdot x+B,
\]
where
\[
B\coloneqq\phi(z)-\phi(y)+\nabla\phi(y)\cdot y-\nabla\phi(z)\cdot z.
\]
Let \(A\coloneqq\nabla\phi(z)-\nabla\phi(y)\). Since \(d\ge3\), for any \(\alpha,\beta\in\mathbb R\) with \(\alpha^2+\beta^2<1\), we may choose \(w\in\operatorname{span}\left\{y,z\right\}^{\perp}\cap\mathbb S^{d-1}\) and define
\(x_{\pm}\coloneqq\alpha y+\beta z\pm\sqrt{1-\alpha^2-\beta^2}w\). The vectors \(x_+\) and \(x_-\) have the same inner products with \(y\) and \(z\). Applying \eqref{eq:a2} to \(x_+\) and \(x_-\), therefore, yields \(A\cdot w=0\). Varying \(w\), we obtain \(A\in\operatorname{span}\left\{y,z\right\}\). Thus, there are \(a,b\in\mathbb R\) such that
\[
f(\alpha)-f(\beta)=a\alpha+b\beta+B \qquad \text{whenever } \alpha^2+\beta^2<1.\]

Taking \(\alpha=\beta=t\) for \(|t|<1/\sqrt2\) gives \(0=(a+b)t+B\), so \(b=-a\) and \(B=0\). Hence,
\[\tag{\(A3\)}\label{eq:a3}f(\alpha)-f(\beta)=a(\alpha-\beta) \qquad \text{whenever } \alpha^2+\beta^2<1.
\]
Taking \(\beta=0\) yields \(f(t)=f(0)+at\) for all \(t\in(-1,1)\). Taking \(x=y\) and \(x=-y\) in \eqref{eq:a3} yields the same formula at \(t=1\) and \(t=-1\). Since \(f(1)=D_\phi(y,y)=0\), we have \(a=-f(0)\). Therefore, setting \(c\coloneqq f(0)\),
\(f(t)=c(1-t)\) for all \(t\in[-1,1]\). Strict convexity implies \(c=f(0)=D_\phi(y,z)>0\), since \(y\neq z\).

Consequently, for all \(u,v\in\mathbb S^{d-1}\), \(D_\phi(u,v)=c(1-u\cdot v)=\frac{c}{2}\norm{u-v}^2\). Finally, for nonzero \(a,b\in\mathbb R^d\), we get
\(D_\phi(\tilde a,\tilde b)=c\left(1-\cos(a,b)\right)\).\end{proof}

\newpage

\section{Online Appendix}

\subsection{Objective Signal Frequencies and Introspection-Proofness}
\label{app:introspection}

Section~\ref{sec:interim} compares \textit{ex-ante} and \textit{ex-interim} participation without specifying an objective signal-generating process. We can sharpen this distinction when there is an objective experiment. Suppose the common prior \(p\) is objective and that signals are generated by a true Blackwell experiment \(m^\ast\colon \Omega\to\Delta(S)\). The true
unconditional distribution of signals is then
\(p_{m^\ast}(s)
=
\sum_{\omega\in\Omega}p(\omega)m^\ast(s\mid\omega)\).

Following \citet{bohren2023behavioral}, a subjective experiment is
\textit{introspection-proof} if it matches the true unconditional distribution of signal
realizations. Formally, define
\[
B(p,m^\ast)
\equiv
\left\{
m\colon \Omega\to\Delta(S) \colon \
\sum_{\omega\in\Omega}p(\omega)m(s\mid\omega)
=
\sum_{\omega\in\Omega}p(\omega)m^\ast(s\mid\omega)
\text{ for every }s\in S
\right\}.
\]
Thus, \(m\in B(p,m^\ast)\) if the agent's subjective model predicts the correct frequency
of each signal, even though it may assign a different meaning to those signals.

This restriction rules out purely prospective disagreement relative to the objective signal
process. If \(m,m'\in B(p,m^\ast)\), then
\[
p_m(s)=p_{m'}(s)=p_{m^\ast}(s)
\qquad\text{for every }s\in S.
\]
As a result, the two agents agree about how often each signal occurs. Any remaining disagreement must concern the posterior interpretation of realized signals: for some signal \(s\), the
agents may have
\(p_m(\cdot\mid s)\neq p_{m'}(\cdot\mid s)\). In the terminology of \citet{bohren2023behavioral}, restricting attention to
\(B(p,m^\ast)\) removes disagreement about signal frequencies and leaves only disagreement
about what signals mean.

The \textit{ex-ante}/\textit{ex-interim} distinction remains meaningful under this restriction.
For any \(m\in B(p,m^\ast)\) and interim surplus vector \(x=(x_s)_{s\in S}\), where \(x_s(\omega)\) is measured relative to the signal-contingent reservation payoff \(\bar u_s\), the \textit{ex-ante} evaluation can be written as
\[
p_m\cdot x
=
\sum_{s\in S}p_m(s)\left(p_m(\cdot\mid s)\cdot x_s\right)
=
\sum_{s\in S}p_{m^\ast}(s)\left(p_m(\cdot\mid s)\cdot x_s\right),
\]
and so within \(B(p,m^\ast)\), \textit{ex-ante} participation aggregates interim gains using
the objective signal frequencies. Model disagreement enters only through the posterior
terms \(p_m(\cdot\mid s)\). In contrast, \textit{ex-interim} participation continues to impose the signal-by-signal constraints \(p_m(\cdot\mid s)\cdot x_s\geq 0\), with \(x_s\) measured relative to \(\bar u_s\). Therefore, even when subjective models are
introspection-proof, \textit{ex-ante} and \textit{ex-interim} participation remain distinct:
the former averages posterior disagreements across signals using the common signal
frequencies, while the latter requires participation after each realized signal separately.

The set \(B(p,m^\ast)\) is, therefore, useful when one wants to focus on interpretive
disagreement while ruling out disagreement about the frequency of evidence. The main text
does not impose this restriction because our baseline comparison treats subjective signal
structures as primitives and does not require an objective experiment. When an objective
experiment is available, however, \(B(p,m^\ast)\) provides a natural refinement of the
\textit{ex-ante}/\textit{ex-interim} taxonomy.

\subsection{What Changes Under a Prior-Optimal Reservation Action}
\label{app:payoff_vectors_decision_problems}

In the main text, we use a fixed reservation payoff \(\bar u\). Under this specification, finite decision problems span the whole Arrow-Debreu surplus space \(\mathbb R^{\Omega\times S}\). This appendix records what changes under the alternative requirement that \(\bar u\) be generated by an action that is prior-optimal under \(p\).

Consider the alternative setup in which the reservation payoff \(\bar u\) must be generated by an action \(a^0\) that is prior-optimal under \(p\), with \(\sum_{\omega\in\Omega}p(\omega)u(a^0,\omega)=\bar u\). Let \(x\) be the surplus vector induced by a decision problem and protocol under this alternative. For every signal \(s\in S\), \(\sum_{\omega\in\Omega}p(\omega)x(\omega,s)\leq0\). Indeed, after signal \(s\), the protocol chooses a feasible mixed action, while the reservation action \(a^0\) is prior-optimal under \(p\). Hence the action chosen after \(s\) cannot generate strictly positive expected surplus under the prior relative to \(\bar u\). This observation identifies the set of payoff vectors implementable via a decision problem with a prior-optimal reservation action. Define
\[
K_p\coloneqq\left\{x\in\mathbb R^{\Omega\times S}\colon\sum_{\omega\in\Omega}p(\omega)x(\omega,s)\leq0\text{ for every }s\in S\right\}.
\]

\begin{lemma}
\label{lem:decision_problem_cone}
A vector \(x\in\mathbb R^{\Omega\times S}\) is generated by some finite decision problem and protocol, with the reservation payoff generated by a prior-optimal action under \(p\), if and only if \(x\in K_p\).
\end{lemma}

\begin{proof}
Necessity follows from the preceding paragraph. For sufficiency, fix \(x\in K_p\). Normalize the reservation payoff to \(\bar u=0\); adding a constant \(\bar u\) to all action payoffs recovers the general case. Construct a decision problem with one reservation action \(a^0\), whose payoff is zero in every state, and one action \(a^s\) for each signal \(s\in S\), with payoff
\[
u(a^s,\omega)=x(\omega,s).
\]
Let the protocol choose \(a^s\) after signal \(s\). Since \(x\in K_p\),
\[
\sum_{\omega\in\Omega}p(\omega)u(a^s,\omega)
=
\sum_{\omega\in\Omega}p(\omega)x(\omega,s)
\le 0
=
\sum_{\omega\in\Omega}p(\omega)u(a^0,\omega)
\]
for every \(s\). Thus \(a^0\) is prior-optimal and represents the reservation payoff \(\bar u=0\). The protocol induces exactly the surplus vector \(x\).
\end{proof}

The lemma shows that the prior-optimal-reservation-action requirement changes the
behavioral domain. Under the fixed-reservation-payoff formulation, finite decision
problems can implement every surplus vector in the Arrow--Debreu space, so the
full-space cone inclusion in \Cref{prop:inclusion_char} is both necessary and sufficient
for the uniform behavioral implication. Once the reservation payoff must instead be
generated by an action that is prior-optimal under \(p\), this spanning argument breaks
down: decision problems implement only the subcone \(K_p\). Hence full-space inclusion
still implies the desired behavioral comparison, because
\[
C(p_{\hat m},p_{\hat m'})\subseteq C(p_m,p_{m'})
\quad\Longrightarrow\quad
C(p_{\hat m},p_{\hat m'})\cap K_p
\subseteq
C(p_m,p_{m'})\cap K_p.
\]
But the converse need not hold. The relevant behavioral implication is now tested only
on surplus vectors in \(K_p\), and two pairs of models may have the required inclusion
after intersecting with \(K_p\) even though their full Arrow--Debreu cones are not
nested. Thus, under the prior-optimal reservation-action restriction, the inclusion order from the main text remains a sufficient condition, but it ceases to be necessary. The exact comparison is instead the restricted cone inclusion
\[
C(p_{\hat m},p_{\hat m'})\cap K_p
\subseteq
C(p_m,p_{m'})\cap K_p,
\]
rather than the full-space inclusion
\(C(p_{\hat m},p_{\hat m'})\subseteq C(p_m,p_{m'})\).

\end{document}